\documentclass[11pt,letterpaper]{article}

\usepackage{typearea}
\typearea{15}

\usepackage{amsmath,amssymb}
\usepackage{amsthm}

\usepackage{xspace}
\usepackage[showonlyrefs,showmanualtags]{mathtools}
\usepackage{bm}
\usepackage{ifpdf}
\usepackage{shadow,shadethm,color}
\usepackage[compact]{titlesec}
\usepackage{algorithm}
\usepackage[noend]{algorithmic}
\numberwithin{algorithm}{section}
\usepackage{subfigure}
\usepackage{verbatim}
\usepackage{times}
\usepackage{paralist}

\usepackage{enumerate}
\usepackage{changepage}
\usepackage{proof}

\definecolor{Darkblue}{rgb}{0,0,0.4}
\definecolor{Brown}{cmyk}{0,0.81,1.,0.60}
\definecolor{Purple}{cmyk}{0.45,0.86,0,0}
\newcommand{\mydriver}{hypertex}
\ifpdf
 \renewcommand{\mydriver}{pdftex}
\fi
\usepackage[breaklinks,\mydriver]{hyperref}
\hypersetup{colorlinks=true,%pdfborder={1 1 1 [3]},%
            citebordercolor={.6 .6 .6},linkbordercolor={.6 .6 .6},%
citecolor=Darkblue,urlcolor=black,linkcolor=Darkblue,pagecolor=black}
\newcommand{\lref}[2][]{\hyperref[#2]{#1~\ref*{#2}}}


\usepackage{enumitem}

\makeatletter
\newcommand{\manuallabel}[2]{\def\@currentlabel{#2}\label{#1}}
\makeatother

\newtheorem{theorem}{Theorem}[section]
\newtheorem{definition}[theorem]{Definition}
\newtheorem{proposition}[theorem]{Proposition}
\newtheorem{conjecture}[theorem]{Conjecture}
\newtheorem{lemma}[theorem]{Lemma}
\newtheorem{fact}[theorem]{Fact}

\newtheorem{claim}[theorem]{Claim}

\newtheorem{corollary}[theorem]{Corollary}
\theoremstyle{definition}

\newcommand{\junk}[1]{}
\newcommand{\ignore}[1]{}

\newcommand{\R}[0]{{\ensuremath{\mathbb{R}}}}

\newcommand{\Z}[0]{{\ensuremath{\mathbb{Z}}}}

\newcommand{\poly}{\operatorname{poly}}

\newcommand{\eps}{\varepsilon}

\newcounter{note}[section]

\newcommand{\initOneLiners}{%
    \setlength{\itemsep}{0pt}
    \setlength{\parsep }{0pt}
    \setlength{\topsep }{0pt}
}

\newcommand{\squishlist}{
 \begin{list}{$\bullet$}
  { \setlength{\itemsep}{0pt}
     \setlength{\parsep}{3pt}
     \setlength{\topsep}{3pt}
     \setlength{\partopsep}{0pt}
     \setlength{\leftmargin}{1.5em}
     \setlength{\labelwidth}{1em}
     \setlength{\labelsep}{0.5em} } }

\newcommand{\squishend}{
  \end{list}  }

\newcommand{\E}{\mathbb{E}}

\newcommand{\calI}{\mathcal{I}}

\newcommand{\calC}{\mathcal{C}}

\newcommand{\CSP}{\mathrm{CSP}}
\newcommand{\SDPOpt}{\mathrm{SDPOpt}}

\newcommand{\Lovasz}{Lov\'{a}sz\xspace}

\newtheorem{innercustomlem}{Lemma}
\newenvironment{customlem}[1]
  {\renewcommand\theinnercustomlem{#1}\innercustomlem}
  {\endinnercustomlem}

\newtheorem{innercustomfact}{Fact}
\newenvironment{customfact}[1]
  {\renewcommand\theinnercustomfact{#1}\innercustomfact}
  {\endinnercustomfact}

\newtheorem{innercustomtheorem}{Theorem}
\newenvironment{customtheorem}[1]
  {\renewcommand\theinnercustomtheorem{#1}\innercustomtheorem}
  {\endinnercustomtheorem}

\newcommand{\SAplus}{$\mathrm{SA}_+$\xspace}
\newcommand{\SAplusop}{\mathrm{SA}_+}
\newcommand{\SA}{$\mathrm{SA}$\xspace}
\newcommand{\SAop}{\mathrm{SA}}

\newcommand{\LSop}{\mathrm{LS}}
\newcommand{\LSplus}{$\mathrm{LS}_+$\xspace}

\newcommand{\SAT}{\mathrm{SAT}}

\newcommand{\wt}[1]{\widetilde{#1}}
\newcommand{\cl}{\mathsf{Cl}}
\newcommand{\suchthat}{\mid}

\newcommand{\indic}[1]{1_{#1}}

\begin{document}

\title{Lower bounds for CSP refutation by SDP hierarchies}
\author{Ryuhei Mori\thanks{Department of Mathematical and Computing Sciences, Tokyo Institute of Technology. \texttt{mori@is.titech.ac.jp}. Supported by MEXT KAKENHI Grant Number 24106008.}
 \and David Witmer\thanks{Computer Science Department, Carnegie Mellon University. \texttt{dwitmer@cs.cmu.edu}. Supported by the National Science Foundation Graduate Research Fellowship Program under grant DGE-1252522, by NSF grants CCF-0747250 and CCF-1116594, and by a CMU Presidential Fellowship.}}
\maketitle

\begin{abstract}
For a $k$-ary predicate $P$, a random instance of CSP$(P)$ with $n$ variables and $m$ constraints is unsatisfiable with high probability when $m \gg n$.  The natural algorithmic task in this regime is \emph{refutation}: finding a proof that a given random instance is unsatisfiable.  Recent work of Allen et al. suggests that the difficulty of refuting CSP$(P)$ using an SDP is determined by a parameter $\mathrm{cmplx}(P)$, the smallest $t$ for which there does not exist a $t$-wise uniform distribution over satisfying assignments to $P$.  In particular they show that random instances of CSP$(P)$ with $m \gg n^{\mathrm{cmplx(P)}/2}$ can be refuted efficiently using an SDP.

In this work, we give evidence that $n^{\mathrm{cmplx}(P)/2}$ constraints are also \emph{necessary} for refutation using SDPs.  Specifically, we show that if $P$ supports a $(t-1)$-wise uniform distribution over satisfying assignments, then the Sherali-Adams$_+$ and Lov\'{a}sz-Schrijver$_+$ SDP hierarchies cannot refute a random instance of CSP$(P)$ in polynomial time for any $m \leq n^{t/2-\eps}$.
 
\end{abstract}

\clearpage

\section{Introduction}

The average-case complexity of constraint satisfaction problems (CSPs) has been studied extensively in computer science, mathematics, and statistical physics.  Despite the vast amount of research that has been done, the hardness of natural algorithmic tasks for random CSPs remains poorly understood.  Given a $k$-ary predicate $P:\{0,1\}^k \to \{0,1\}$, we consider random instances of CSP($P$) with $n$ variables and $m$ constraints.  Each constraint is chosen independently and consists of $P$ applied to $k$ literals (variable or their negations) chosen independently and uniformly at random.  Whether or not a random CSP is satisfiable depends on its clause density $\frac{m}{n}$.  It is conjectured that for any nontrivial CSP there is a \emph{satisfiability threshold} $\alpha(P)$ depending on the choice of predicate $P$: For $m < \alpha(P) \cdot n$, an instance is satisfiable with high probability, and $m > \alpha(P) \cdot n$, an instance is unsatisfiable with high probability.  This conjecture has been proven in the case of $k$-$\SAT$ for large enough $k$ \cite{DSS15}, but, to our knowledge, remains open for all other predicates.  Even so, it is easy to show that when $m \gg n$, an instance is unsatisfiable with high probability.
In the low density, satisfiable regime, the major research goal is to develop algorithms that find satisfying assignments.  In the high density, unsatisfiable regime, the goal is to \emph{refute} an instance, i.e., find a short certificate that there is no solution.

In this paper, we study refutation.   A refutation algorithm takes a random instance $\calI$ of $\CSP(P)$ and returns either ``unsatisfiable" or "don't know".  It must satisfy two conditions: (1) it is never wrong, i.e., if $\calI$ is satisfiable, it must return ``don't know" and (2) it returns ``unsatisfiable" with high probability over the choice of the instance.  As $m$ increases, refutation becomes easier.  The objective, then, is to refute instances with $m$ as small as possible.  This problem has been studied extensively and is related to hardness of approximation \cite{Fei02}, proof complexity \cite{BB02}, statistical physics \cite{CLP02}, cryptography \cite{ABW10}, and learning theory \cite{DLS14}.  Much research has focused on finding algorithms for refutation, especially in the special case of SAT; see \cite{AOW15} for references.

Most known refutation algorithms are based on semidefinite programming (SDP).   For now, we think of an SDP relaxation of an instance $\calI$ of $\CSP(P)$ as a black box that returns a number $\SDPOpt \in [0,1]$ that upper bounds the maximum fraction of constraints that can be simultaneously satisfied.  An SDP-based refutation algorithm takes a random instance $\calI$ of $\CSP(P)$, solves some SDP relaxation of $\calI$, and return ``Unsatisfiable" if and only if $\SDPOpt < 1$.  Many known polynomial-time algorithms for refutation fit into this framework (e.g., \cite{AOW15,BM15,FGK05,COGL04,FO04}).  It is then natural to ask the following question.
\begin{quote}
What is the minimum number of constraints needed to refute random instances of CSP$(P)$ using an efficient SDP-based refutation algorithm?
\end{quote}
Allen et al. give an upper bound on the number of constraints required to refute an instance of $\CSP(P)$ in terms of a parameter $\mathrm{cmplx}(P)$, defined to be the minimum $t$ such that there is no $t$-wise uniform distribution over satisfying assignments to $P$ \cite{AOW15}.
\begin{theorem}[\textup{\cite{AOW15}}]
There is an efficient SDP-based algorithm that refutes a random instance $\calI$ of $\CSP(P)$ with high probability when $m \gg n^{\mathrm{cmplx}(P)/2}$.
\end{theorem}
Clearly, $1 \leq \mathrm{cmplx}(P) \leq k$ for nontrivial predicates.  Also, $\mathrm{cmplx}(P) = k$ when $P$ is $k$-XOR or $k$-SAT.

For special classes of predicates, we know that  $n^{\mathrm{cmplx}(P)/2}$ constraints are also necessary for refutation by SDP-based algorithms.  Schoenebeck considered arity-$k$ predicates $P$ whose satisfying assignments are a superset of $k$-XOR's; these include $k$-SAT and $k$-XOR.  For such predicates, he showed that polynomial-size sum of squares (SOS) SDP relaxations cannot refute random instances with $m \leq n^{k/2-\eps}$ \cite{Sch08} using a proof previously discovered by Grigoriev \cite{Gri01}.  Lee, Raghavendra, and Steurer showed that the SOS relaxation of $\CSP(P)$ is at least as powerful as an arbitrary SDP relaxation of comparable size \cite{LRS15}.  With Schoenebeck's result, this implies that no polynomial-size SDP can be used to refute random instances of $k$-XOR or $k$-SAT when $m \leq n^{k/2-\eps}$.  This leads us to make the following conjecture.
\begin{conjecture} \label{conj:lb}
Let $\eps$ be a constant greater than $0$.  Given a random instance $\calI$ of $\CSP(P)$ with $m \leq n^{\mathrm{cmplx}(P)/2-\eps}$, with high probability any polynomial-size SDP relaxation of $\calI$ has optimal value $1$ and can therefore not be used to refute $\calI$.
\end{conjecture}
It suffices to prove this conjecture for SOS SDP relaxations \cite{LRS15}.  This would essentially complete our understanding of the power of SDP-based refutation algorithms.  Prior to this work, this SOS version Conjecture~\ref{conj:lb} appeared in \cite{AOW15}; we know of no other mention of this conjecture in the literature.\footnote{Barak, Kindler, and Steurer \cite{BKS13} made a related but different conjecture that the basic SDP relaxation is optimal for random CSPs.} 

Note that in the special case of $k$-XOR, SDP algorithms are \emph{not} optimal.  A $k$-XOR constraint is a linear equation mod $2$, so Gaussian elimination can be used to refute \emph{any} unsatisfiable $k$-XOR instance.  A random $k$-XOR instance with $m \gg n$ can therefore be refuted with high probability.  SDP-based algorithms, on the other hand, require $m \geq n^{k/2}$ \cite{Gri01,Sch08,LRS15}.  More generally, if $P$ can be written as a degree-$d$ polynomial over $\mathbb{F}_2$, then Gaussian elimination can be used to refute random instances with $m = O(n^d)$.  See \cite{OW14} for more details.

Some partial progress has been made toward proving Conjecture~\ref{conj:lb}.  Building on results of Benabbas et al. \cite{BGMT12} and Tulsiani and Worah \cite{TW13}, O'Donnell and Witmer proved lower bounds for the Sherali-Adams (\SA) linear programming (LP) hierarchy and the Sherali-Adams$_+$ (\SAplus) and \Lovasz-Schrijver$_+$ (\LSplus) SDP hierarchies.  All three of these hierarchies are weaker than SOS.  The \SAplus hierarchy gives an optimal approximation for any CSP in the worst case assuming the Unique Games Conjecture \cite{Rag08}.  They showed that polynomial-size Sherali-Adams linear programming (LP) relaxations cannot refute random instances with $m \leq n^{\mathrm{cmplx}(P)/2-\eps}$ \cite{OW14}; this implies that no polynomial-size LP can refute random instances with with $m \leq n^{\mathrm{cmplx}(P)/2-\eps}$ by work of Chan et al. \cite{CLRS13}  They also showed that \SAplus cannot refute random instances with $m \leq n^{\mathrm{cmplx}(P)/2-\eps}$ when a set of $o(m)$ constraints has been removed \cite{OW14}.  Also, they proved that \SAplus and \LSplus cannot refute fully random instances with $m \leq n^{\mathrm{cmplx}(P)/2-1/3-\eps}$.  Much less is known for SOS.  For predicates $P$ that support a pairwise uniform distribution over satisfying assignments, Barak, Chan, and Kothari showed that polynomial-size SOS relaxations cannot refute random instances with $m = \Omega(n)$ in which $o(m)$ constraints have been removed \cite{BCK15}.

In addition, there is a long history of related work on lower bounds for refutation in proof complexity (e.g., \cite{Gri01, KI06, BOG+06, TW13}).  Specifically, \SA, \SAplus, \LSplus, and SOS have corresponding static semialgebraic proof systems and proving integrality gaps for these LP and SDP relaxations in equivalent to proving rank or degree lower bounds for refutations in these proof systems.

\paragraph{Results} Our contribution is two-fold: First, we remove the assumption that a small number of constraints are deleted to show that fully-random CSP instances have integrality gaps in \SAplus for $m \leq n^{t/2-\eps}$.  As in \cite{BGMT12} and \cite{OW14}, we prove this result for predicates over possibly larger alphabets $[q]$.
\begin{theorem} \label{thm:sa+-intro}
Let $P:[q]^k \to \{0,1\}$ be $(t-1)$-wise uniform-supporting and let $\mathcal{I}$ be a random instance of $\CSP(P)$ with $n$ variables and $m \leq n^{t/2 - \eps}$ constraints.  Then with high probability the \SAplus relaxation for $\mathcal{I}$ has value $1$, even after $\Omega(n^{\frac{\eps}{t-2}})$ rounds.
\end{theorem}

Second, we use this result to show that fully random instances have \LSplus integrality gaps for $m \leq n^{t/2-\eps}$.  Recall that \LSplus gives relaxations of $0$/$1$-valued integer programs, so we restrict our attention here to Boolean CSPs with $P:\{0,1\}^k \to \{0,1\}$.
\begin{theorem} \label{thm:ls+-intro}
Let $P:\{0,1\}^k \to \{0,1\}$ be $(t-1)$-wise uniform-supporting and let $\mathcal{I}$ be a random instance of $\CSP(P)$ with $n$ variables and $m \leq n^{t/2 - \eps}$ constraints.  Then with high probability the \LSplus relaxation for $\mathcal{I}$ has value $1$, even after $\Omega(n^{\frac{\eps}{t-2}})$ rounds.
\end{theorem}

In their strongest form, our results hold for a static variant of the \LSplus SDP hierarchy that is at least as strong as both \SAplus and \LSplus and is dual to the static \LSplus proof system studied in previous work (e.g., \cite{GHP02, KI06}).  We define this static \LSplus hierarchy in Section~\ref{sec:prelims}.
\begin{theorem} \label{thm:static-ls+-intro}
Let $P:[q]^k \to \{0,1\}$ be $(t-1)$-wise uniform-supporting and let $\mathcal{I}$ be a random instance of $\CSP(P)$ with $n$ variables and $m \leq n^{t/2 - \eps}$ constraints.  Then with high probability the static \LSplus relaxation for $\mathcal{I}$ has value $1$, even after $\Omega(n^{\frac{\eps}{t-2}})$ rounds.
\end{theorem}
Tulsiani and Worah proved this theorem in the special case of pairwise independence and $O(n)$ constraints \cite[Theorem 3.27]{TW13}.\footnote{Actually, \cite{TW13} prove a rank lower bound for the dual static \LSplus proof system, but this is equivalent to a rank lower bound for the static \LSplus SDP hierarchy we consider here.} 
These results provide further evidence for Conjecture~\ref{conj:lb} and, in particular, give the first examples of SDP hierarchies that are unable to refute CSPs with $(t-1)$-wise uniform-supporting predicates when $m \leq n^{t/2 - \eps}$.

From a dual point of view, we can think of \SAplus, \LSplus, and static \LSplus as semialgebraic proof systems and our results can be equivalently stated as rank or degree lower bounds for these proof systems.
\begin{theorem} \label{thm:proof-systems-intro}
Let $P:\{0,1\}^k \to \{0,1\}$ be $(t-1)$-wise uniform-supporting and let $\mathcal{I}$ be a random instance of $\CSP(P)$ with $n$ variables and $m \leq n^{t/2 - \eps}$ constraints.  Then with high probability any \SAplus or static \LSplus refutation of $\mathcal{I}$ requires degree $\Omega(n^{\frac{\eps}{t-2}})$ and any \LSplus refutation of $\mathcal{I}$ requires rank $\Omega(n^{\frac{\eps}{t-2}})$.
\end{theorem}

In another line of work, Feldman, Perkins, and Vempala \cite{FPV15} showed that if a predicate $P$ is $(t-1)$-wise uniform supporting, then any statistical algorithm based on an oracle taking $L$ values requires $m = \wt{\Omega}(n^t/L)$ to refute.  They further show that the dimension of any convex program refuting such a CSP must be at least $\wt{\Omega}(n^{t/2})$.  These lower bounds are incomparable to the integrality gap results stated above: While statistical algorithms and arbitrary convex relaxations are more general computational models, standard SDP hierarchy relaxations for $k$-CSPs, including the \SAplus and \LSplus relaxations we study, have dimension $n^{O(k)}$ and are therefore not ruled out by this work.

\paragraph{Techniques}
To solve $\CSP(P)$ exactly, it suffices to optimize the expected fraction of satisfied constraints over distributions on assignments $[q]^n$.  This, of course, is hard, so relaxations like \SA, \SAplus, \LSplus, and SOS instead optimize over local distributions on assignments to smaller sets of variables .  As the number of rounds of the relaxation increases, we look at distributions over assignments to larger and larger sets of variables; the $r$-round relaxation considers distributions over assignments to sets of size at most $r$ and has size $n^{O(r)}$.  The $r$-round \SAplus relaxation requires that (1) local distributions on assignments to sets of at most $r$ variables satisfy consistency conditions and (2) the covariance matrix corresponding to these local distributions is positive semidefinite (PSD); see Section~\ref{sec:prelims} for precise definitions.  We know that when $m \leq n^{\mathrm{cmplx}(P)/2-\eps}$, there exist local distributions satisfying (1) \cite{BGMT12, OW14}.

When the number of constraints is $O(n)$, previous work showed that the covariance matrix corresponding to the \cite{BGMT12} local distributions is PSD by proving that it is diagonal and has nonnegative entries.  If the covariance matrix is diagonal, then there is no correlation between assignments to pairs of variables under the corresponding local distributions.  This condition holds for instances with number of constraints small enough, but correlations between variables arise as the number of constraints increases.

We prove PSDness in the presence of these correlations by showing that they must remain local.  Our argument extends a technique of Tulsiani and Worah \cite{TW13}.  We prove that the graph induced by correlations between variables must have small connected components, each of which corresponds to a small block of nonzero entries in the covariance matrix.  Since these blocks are small, condition (1) guarantees that for each block there exists an actual distribution on assignments to the variables of that block.  This means that each of these blocks is the covariance matrix of an actual distribution and must therefore be PSD.  The entire covariance matrix must then be PSD: It can be written as a block diagonal matrix in which each block is PSD.

\section{Preliminaries} \label{sec:prelims}

\subsection{Constraint satisfaction problems}
\begin{definition}
Given a predicate $P:[q]^k \to \{0,1\}$, an instance $\calI$ of the CSP$(P)$ problem with variables $x_1,\ldots,x_n$ is a multiset of $P$-constraints.  Each $P$-constraint is a tuple $(c,S)$, where $c \in [q]^k$ is the negation pattern and $S \in [n]^k$ is the scope.  The corresponding constraint is $P(x_S + c) = 1$, where $x_S = (x_i)_{i \in S}$ and $+$ is component-wise addition mod $q$.

In the decision version of CSP$(P)$, we want to determine whether or not there exists an assignment satisfying all constraints of a given instance $\calI$.  In the optimization version, the objective is to maximize the fraction of simultaneously satisfied constraints.  That is, we define $\mathrm{Val}_{\calI}(x) = \frac{1}{m}\sum_{(c,S) \in \calI} P(c + x_S)$ and want to find $x \in [q]^n$ maximizing  $\mathrm{Val}_{\calI}(x)$.  We will write $\max_{x} \mathrm{Val}_{\calI}(x)$ as $\mathrm{Opt(\calI)}$.
\end{definition}

We will show that SOS cannot solve the decision version for random instances $\calI$ with small enough number of constraints even though such instances are far from satisfiable.  This implies that SOS cannot show that $\mathrm{Opt(\calI)} < 1$ for such instances.

Next, we define our random model.  We consider instances in which $m$ constraints are drawn independently and uniformly at random from among all $q^k n^k$ possible constraints with replacement.  We distinguish between different orderings of the scope, as $P$ may not be symmetric.  The specific details of this definition are not important; our results hold for any similar model.  For example, see \cite[Appendix~D]{AOW15}.  A random instance is likely to be highly unsatisfiable: It is easy to show that $\mathrm{Opt(\calI)} = \frac{|P^{-1}(1)|}{q^k} + o(1)$ for $m \geq n \log n$ with high probability.

Given an instance $\calI$, the associated $k$-uniform hypergraph $H_{\calI}$ on $V = [n]$ has a hyperedge $S$ if and only if $S$ is the scope of some constraint of $\calI$.  Here, we disregard the orderings of the constraint scopes.  Given a hypergraph $H$ and a subset of vertices $T$, we let $H[T]$ be the subhypergraph induced by $T$.

We consider predicates for which there exist distributions over satisfying assignments that look uniform on every small enough set of bits.  Formally, we study the following condition.
\begin{definition}
A predicate $P:[q]^k \to \{0,1\}$ is $t$-wise uniform supporting if there exists a distribution $\mu$ over $[q]^k$ supported on $P$'s satisfying assignments such that $\Pr_{z \sim \mu}[z_T = \alpha] = q^{-|T|}$ for all $\alpha \in [q]^{|T|}$ and for all $T \subseteq [k]$ with $1 \leq |T| \leq t$.
\end{definition}

\subsection{LP and SDP hierarchies}

\subsubsection{Representing $\CSP(P)$ with polynomial inequalities}
An LP or SDP hierarchy is a procedure for constructing increasingly tight relaxations of a set of polynomial inequalities.  In our case, these polynomial inequalities represent the constraints of a CSP.  In this section, we describe two standard ways of writing these polynomial formulations.  Both encode the decision version of CSP$(P)$.  The relaxations based on these encodings are at least as strong as corresponding relaxations of the maximization version of CSP$(P)$, so our lower bounds for the decision problem imply lower bounds for the maximization problem.  We describe the case of binary alphabets first and then mention how CSPs with larger alphabets can be encoded using binary variables.

For \SA, \SAplus, and static \LSplus relaxations, we represent each constraint as a degree-$k$ polynomial equality.  Let $P'(x)$ be the unique multilinear degree-$k$ polynomial such that $P'(z) = P(z)$ for all $z \in \{0,1\}^k$.  Also, given $a \in [0,1]^k$ and $b \in \{0,1\}$, use $a^{(b)}$ to denote $a$ if $b = 0$ and $1-a$ if $b=1$.  For $z \in [0,1]^k$ and $c \in \{0,1\}^k$, let $z^{(c)} \in [0,1]^k$ be such that $(z^{(c)})_i = z_i^{(c_i)}$.  The degree-$k$ formulation of $\calI$ is defined as follows.
\begin{equation} \label{eq:deg-k-relax}
R_{\calI} = \{P'(x_S^{(c)}) = 1 \suchthat (c,S) \in \calI\}.
\end{equation}

For \LSplus, on the other hand, we have to start with a set of \emph{linear} inequalities and only consider the binary alphabet case.  Recall that any nontrivial arity-$k$ predicate $P$ can be represented as a conjunction of at most $2^k - 1$ disjunctions of arity $k$.  In particular, letting $F = \{z \in \{0,1\}^k \suchthat P(z)=0\}$, we see that
\begin{equation} \label{eq:cnf-form}
P(z) = \bigwedge_{f \in F} \bigvee_{i = 1}^k f_i \oplus z_i.
\end{equation}
Using \eqref{eq:cnf-form}, we can represent $\calI$ as a $k$-SAT instance with at most $(2^k - 1) \cdot m$ constraints and then consider the standard linear relaxation of this $k$-$\SAT$ instance.  For each clause $\bigvee_{i = 1}^k c_i \oplus z_i$, we add the inequality $\sum_{i=1}^k z_i^{(c_i)} \geq 1$ and obtain the following linear relaxation for $\calI$.
\begin{equation} \label{eq:lin-relax}
L_{\calI} = \left\{\sum_{i = 1}^k x_{S_i}^{(c_i \oplus f_i)} \geq 1~\middle|~ (c,S) \in \calI, f \in F\right\}.
\end{equation}

It is more natural to apply \SA, \SAplus, and static \LSplus to \eqref{eq:deg-k-relax}, but applying \SA, \SAplus, and static \LSplus to \eqref{eq:lin-relax} yields relaxations that are approximately equivalent.
\begin{lemma} \label{lem:diff-encodings}
Let $r \geq k$ and let $\calI$ be an instance of CSP$(P)$ with binary alphabet.  Then the following statements hold.
\begin{enumerate}
\item $\SAop^r(R_{\calI}) \subseteq \SAop^{r+k+1}(L_{\calI})$ and $\SAop^r(L_{\calI}) \subseteq \SAop^{r+k+1}(R_{\calI})$.
\item $\SAplusop^r(R_{\calI}) \subseteq \SAplusop^{r+k+1}(L_{\calI})$ and $\SAplusop^r(L_{\calI}) \subseteq \SAplusop^{r+k+1}(R_{\calI})$.
\item $\mathrm{StaticLS}_+^r(R_{\calI}) \subseteq \mathrm{StaticLS}_+^{r+k+1}(L_{\calI})$ and $\mathrm{StaticLS}_+^r(L_{\calI}) \subseteq \mathrm{StaticLS}_+^{r+k+1}(R_{\calI})$.
\end{enumerate}
\end{lemma}
We include the proof in Appendix~\ref{sec:diff-encodings}.  We define the $\SAop^r$, $\SAplusop^r$, and $\mathrm{StaticLS}_+^r$ operators in the next sections.

\paragraph{Larger alphabets}
We can also consider CSPs with variables taking values in $[q]$.  We use $\{0,1\}$-valued variables $x_{i,a}$ such that $x_{i,a} = 1$ if and only if variable $i$ is assigned value $a$.  Given predicate $P:[q]^k \to \{0,1\}$, let $P_{01}:\{0,1\}^{qk} \to \{0,1\}$ be a polynomial in variables $z_{i,a}$ such that
\[
P_{01}(z) = \sum_{\alpha \in \Z_q^k} P(\alpha) \prod_{i \in [k]} z_{i,\alpha_i}.
\]
Observe that $P(\alpha) = P_{01}(z)$ if $z_{i,\alpha_i} = 1$ and $z_{i,b} = 0$ for all $i \in [k]$ and $b \ne \alpha_i$.
Given $z \in \R^{qk}$ and $c \in [q]^k$, define $z^{(c)} \in \R^{qk}$ so that $z^{(c)}_{i,a} = z_{i,(a+c_i)\,\mathrm{mod}\,q}$.
The constraints in this formulation have degree $k$.

We can encode the decision problem as the following system of polynomial inequalities.
\begin{equation*} \label{eq:deg-k-relax-q}
R_{\calI} = \left\{P_{01}(x_{S \times [q]}^{(c)}) = 1\,\middle|\, (c,S) \in \calI\right\}.
\end{equation*}

\subsubsection{Sherali-Adams} The Sherali-Adams (\SA) linear programming hierarchy gives a family of locally consistent distributions on assignments to sets of variables.  As the size of these sets increases, the relaxation becomes tighter.
\begin{definition}
Let $\{D_S\}$ be a family of distributions $D_S$ over $[q]^{S}$ for all $S \subseteq [n]$ with $|S| \leq r$.  We say that $\{D_S\}$ is $r$-locally consistent if for all $T \subseteq S \subseteq [n]$ with $|S| \leq r$, the marginal of $D_S$ on $T$ is equal to $D_T$.  In symbols, we write this condition as $D_{T}(T = \alpha) = D_S(T = \alpha)$, where
\[
D_{S}(T = \alpha) = \sum_{\substack{\beta \in [q]^S \\ \beta_T = \alpha}} D_S(\beta).
\]
\end{definition}

We can now define the \SA relaxation of a set of constraints.  We will extend the distributions $\{D_S\}$ to distributions over assignments to $x$ in $[q]^n$ by choosing assignments to $x_{[n] \setminus S}$ uniformly at random.  Given a family of $r$-locally consistent distributions $\{D_S\}$ and a monomial $x^T = \prod_{i \in T} x_i$ with $|T| \leq r$, we define $\E_D[x^T] = \E_{D_T}[x^T]$.  We extend this definition to degree-$r$ polynomials by linearity.  For $T \subseteq [n]$ and $\alpha \in [q]^{|T|}$, let $\indic{\{x_T = \alpha\}}(x)$ be the indicator polynomial for the event $x_T = \alpha$.  Let $\deg(\cdot)$ denote the multilinear degree of a polynomial, i.e., the degree after replacing all appearances of $x_i^2$ with $x_i$.
\begin{definition}
Let $A = \{g_1(x) \geq 0,  g_2(x) \geq 0, \ldots, g_m(x) \geq 0\}$ be a set of polynomial constraints such that for all $g \in A$, $\deg(g) \leq r$.  We define the $r$-round \SA relaxation for $A$ to be the set of all families of distributions $\{D_S\}_{S\subseteq [n],\, |S|\le r}$ such that $D_S$ is a distribution over $[q]^{S}$ satisfying the following two properties.
\begin{enumerate}
\item $\{D_S\}_{S\subseteq [n],\, |S|\le r}$ is $r$-locally consistent.
\item $\E_{x \sim D}[\indic{\{x_T = \alpha\}}(x) \cdot g(x)] \geq 0$ for all $g \in A$, $T \subseteq [n]$, and $\alpha \in [q]^{|T|}$ such that $\deg(\indic{\{x_T = \alpha\}}(x) \cdot f(x)) \leq r$. \label{enum:sa-sat}
\end{enumerate}
We denote this set of families of distributions as $\SAop^r(A)$.  $\SAop^r(A)$ is a polytope of size $n^{O(r)}$ and we can therefore check feasibility in time $n^{O(r)}$.  We point out that $\SAop^{r+1}(A) \subseteq \SAop^r(A)$ and $\SAop^n(A)$ exactly captures feasibility of $A$ over $[q]^n$.
\end{definition}

Specialized to an instance $\calI$ of CSP$(P)$, we consider $\SAop^r(R_{\calI})$ and write Condition~\ref{enum:sa-sat} as
\[
\E_{x \sim D}[\indic{\{x_T = \alpha\}}(x) \cdot (P(x_S + c) - 1)] = 0
\]
for all $(c,S) \in \calI$, all $T \subseteq [n]$ such that $\deg(\indic{\{x_T = \alpha\}}(x) \cdot (P(x_S + c) - 1)) \leq r$, and all $\alpha \in [q]^{|T|}$.

We will only consider the stronger feasibility formulations of relaxations (rather than optimization versions) because lower bounds for these feasibility formulations immediately imply lower bounds for the corresponding optimization versions.

For larger alphabets, implementing this definition as a linear program requires requires writing the constraints as polynomials in Boolean variables.  We therefore identify $\SAop^r(A)$ with $\SAop^r(A_{01})$, where the $A_{01}$ is the inequalities of $A$ written as polynomials in the Boolean variables $x_{i,a}$ as described above.  We make this same identification for \SAplus and static \LSplus below.

In previous work, O'Donnell and Witmer \cite{OW14} extended a theorem of Benabbas et al. \cite{BGMT12} to obtain a lower bound for \SA relaxations of CSP($P$) with $m = n^{t/2-\eps}$ when $P$ is $(t-1)$-wise uniform supporting.
\begin{theorem}[\textup{\cite{BGMT12, OW14}}]
Let $P:[q]^k \to \{0,1\}$ be $(t-1)$-wise uniform-supporting and let $\mathcal{I}$ be a random instance of $\CSP(P)$ with $n$ variables and $m \leq n^{t/2 - \eps}$ constraints.  Then with high probability $\SAop^r(R_{\calI})$ is nonempty for $r=\Omega(n^{\frac{\eps}{t-2}})$ rounds.
\end{theorem}
As mentioned above, feasibility of the \SA relaxation of the decision version of CSP$(P)$ immediately implies that the optimization version of the \SA relaxation has value $1$, i.e., \SA thinks all constraints can be satisfied.  The same holds for the other relaxations we consider; we only look at feasibility for the rest of the paper.

\subsubsection{Sherali-Adams$_+$} The Sherali-Adams$_+$ (\SAplus) SDP hierarchy additionally requires the second moment matrix of these distributions to be PSD.  Given a family of local distributions $\{D_S\}$, define $M = M(D) \in \R^{(nq+1) \times (nq+1)}$ to be the symmetric matrix indexed by $(0, [n] \times [q])$ such that
\begin{align*}
M(0,0) &= 1\\
M(0,(i,a)) &= D_{\{i\}}(x_i = a) \\
M((i,a),(j,b)) &= D_{\{i,j\}}(x_i = a \wedge x_j = b).
\end{align*}
Note that the $((i,a), (i,b))$-element of $M$ is $D_{\{i\}}(x_i = a)$ if $a=b$ and is 0 if $a\ne b$.
\begin{definition}
Given a set of constraints $A$ as above, we define $\SAplusop^r(A)$ to be the set of all families of distributions $\{D_S\}_{S\subseteq [n],\, |S|\le r}$ over $[q]^{S}$ in $\SAop^r(A)$ satisfying the following additional condition.
\begin{enumerate}
\item[3.] $M$ is PSD.
\end{enumerate}
\end{definition}

We can equivalently define \SAplus by requiring the covariance matrix of the locally consistent $\{D_S\}$ distributions to be positive semidefinite.
\begin{definition}
Given $r$-locally consistent distributions $\{D_S\}$ with $r \ge 2$, the covariance matrix $\Sigma = \Sigma(D)$ is defined to be
\[
\Sigma((i,a),(j,b)) = D_{\{i,j\}}(x_i = a \wedge x_j = b) - D_{\{i\}}(x_i = a) \cdot D_{\{j\}}(x_j = b).
\]
\end{definition}

These two formulations are equivalent \cite{WJ08}.
\begin{lemma} \label{lem:sa+-covariance}
M is PSD if and only if $\Sigma$ is PSD.
\end{lemma}

We include the proof in Appendix~\ref{apdx:Schur}.  The covariance matrix condition will be more convenient for us to work with.  For an instance $\calI$ of CSP($P$), we will consider feasibility of $\SAplusop^r(R_{\calI})$.  We prove the following theorem.
\begin{customtheorem}{\ref{thm:sa+-intro}}[restated]
Let $P:[q]^k \to \{0,1\}$ be $(t-1)$-wise uniform-supporting and let $\mathcal{I}$ be a random instance of $\CSP(P)$ with $n$ variables and $m \leq n^{t/2 - \eps}$ constraints.  Then with high probability $\SAplusop^r(R_{\calI})$ is nonempty for $r=\Omega(n^{\frac{\eps}{t-2}})$ rounds.
\end{customtheorem}

\subsubsection{\Lovasz-Schrijver$_+$} We now define the \Lovasz-Schrijver$_+$ (\LSplus) SDP relaxation for problems whose variables are $0/1$-valued.  Given an initial polytope $K \in \R^n$, we would like to generate a sequence of progressively tighter relaxations.  To define one \LSplus lift-and-project step, we will use the cone
\[
\wt{K} = \{(\lambda, \lambda x_1, \ldots, \lambda x_n)~|~\lambda > 0, x_1,\ldots,x_n \in K\}.
\]
$K$ can be recovered by taking the intersection with $x_0 = 1$.
\begin{definition} \label{def:ls+}
Let $\wt{K}$ be a convex cone in $\R^{n+1}$.  Then the lifted \LSplus cone $N_+(\wt{K})$ is the cone of all $y \in \R^{n+1}$ for which there exists an $(n+1) \times (n+1)$ matrix $Y$ satisfying the following:
\begin{enumerate}
\item $Y$ is symmetric and positive semidefinite.
\item For all $i$, $Y_{ii} = Y_{i0} = y_i$.
\item For all $i$, $Y_i \in \wt{K}$ and $Y_0-Y_i \in \wt{K}$
\end{enumerate}
where $Y_i$ is the $i$th column of $Y$.  Then we define $N_+(K)$ to be $N_+(\wt{K}) \cap \{x_0 = 1\}$.  The $r$-round \LSplus relaxation of a polytope $K$ results from applying the $N_+$ operator $r$ times.  That is, we define $N_+^r(K) = N_+(N_+^{r-1}(K))$.  $Y$ is called a protection matrix for $y$.  A solution to the $r$-round \LSplus relaxation for a polytope $K \in \R^n$ defined by $\poly(n)$ linear constraints can be computed in time $n^{O(r)}$ using an SDP.
\end{definition}

For an instance $\calI$ of CSP($P$), we will consider feasibility of $N_+^r(L_{\calI})$.
\begin{customtheorem}{\ref{thm:ls+-intro}}[restated]
Let $P:[q]^k \to \{0,1\}$ be $(t-1)$-wise uniform-supporting and let $\mathcal{I}$ be a random instance of $\CSP(P)$ with $n$ variables and $m \leq n^{t/2 - \eps}$ constraints.  Then with high probability $N_+^r(L_{\calI})$ is nonempty for $r=\Omega(n^{\frac{\eps}{t-2}})$ rounds.
\end{customtheorem}

\subsubsection{Static \LSplus} The static \LSplus relaxation strengthens both \SAplus and \LSplus.  As in the case of \SAplus, we start with a family of $r$-locally consistent distributions and then further require that they satisfy certain positive semidefiniteness constraints.  In particular, for all $X \subseteq [n]$ with $|X| \leq r-2$ and all $\alpha \in [q]^{X}$, define the matrices $M_{X,\alpha} = M_{X,\alpha}(D) \in \R^{(nq+1) \times (nq+1)}$ as follows.
\begin{align*}
M_{X,\alpha}(0,0) &= D_{X}(X = \alpha)\\
M_{X,\alpha}(0,(i,a)) &= D_{\{i\} \cup X}(x_i = a \wedge X = \alpha) \\
M_{X,\alpha}((i,a),(j,b)) &= D_{\{i,j\} \cup X}(x_i = a \wedge x_j = b \wedge X = \alpha).
\end{align*}
In addition to the \SA constraints, the $r$-round static \LSplus relaxation $\mathrm{StaticLS}_+^r(F)$ satisfies the following constraint.

\begin{definition}
Given a set of constraints $A$ as above, we define $\mathrm{StaticLS}_+^r(A)$ to be the set of all families of distributions $\{D_S\}_{S\subseteq [n],\, |S|\le r}$ over $[q]^{S}$ in $\SAop^r(A)$ satisfying the following additional condition.
\begin{enumerate}
\item[3$'$.] $M_{X,\alpha}$ is PSD for all $X \subseteq [n]$ with $|X| \leq r-2$ and all $\alpha \in [q]^{X}$.
\end{enumerate}
\end{definition}

Observe that these positive semidefiniteness constraints can be formulated as a positive semidefiniteness constraint for a single matrix.  In particular, let $\mathcal{M}$ be the block diagonal matrix with the $M_{X,\alpha}$'s on the diagonal.  Then $\mathcal{M}$ has size at most $(qn)^{O(r)}$ and is PSD if and only if all of the $M_{X,\alpha}$'s are PSD.  Unlike \LSplus, this hierarchy easily generalizes to non-binary alphabets.

Alternatively, we can think of this hierarchy as requiring covariance matrices of conditional distributions to be positive semidefinite.  Given a set of local distributions $\{D_S\}$, a set of variables $X \subseteq [n]$, and an assignment $\alpha \in \{0,1\}^{X}$ such that $D_X(X=\alpha) > 0$, define a set of conditional local distributions $\{D_{S|x=\alpha}\}$ by $D_{S|X=\alpha}(\beta) = \frac{D_{S \cup X}(S = \beta \wedge X = \alpha)}{D_X(X = \alpha)}$.  Tulsiani and Worah showed that if the initial family of local distributions is $r$-locally consistent, then the corresponding family of conditional distributions will be $(r-|X|)$-locally consistent \cite{TW13}.
\begin{lemma}[\textup{\cite[Lemma~3.13]{TW13}}] \label{lem:cond-local-cons}
Let $X \subseteq [n]$ and let $\{D_S\}$ be a family of $r$-locally consistent distributions for sets $S \subseteq [n]$ such that $S \cap X = \emptyset$ and $|S \cup X| \leq r$.   Then the family of conditional distributions $\{D_{S| X = \alpha}\}$ is $(r-|X|)$-locally consistent for any $\alpha \in \{0,1\}^X$ such that $D_X(X=\alpha) > 0$.
\end{lemma}
We include the proof of this lemma in Appendix~\ref{sec:cond-local-cons-pf}.

Given such a family of conditional local consistent distributions, the conditional covariance matrix $\Sigma_{X,\alpha}$ is defined as follows.
\begin{definition}
Given $X \subseteq [n]$, $\alpha \in \{0,1\}^X$, and $r$-locally consistent conditional distributions $\{D_{S | X=\alpha}\}$ with $r \geq 2$, the conditional covariance matrix $\Sigma_{X,\alpha} = \Sigma_{X,\alpha}(D)$ is defined to be
\[
\Sigma_{X,\alpha}((i,a),(j,b)) = \begin{cases} D_{\{i,j\}|X=\alpha}(x_i = a \wedge x_j = b) - D_{\{i\}|X=\alpha}(x_i = a) \cdot D_{\{j\}|X=\alpha}(x_j = b) & \text{if $D_X(\alpha) > 0$} \\
0 & \text{otherwise.}
\end{cases}
\]
\end{definition}
Lemma~\ref{lem:sa+-covariance} generalizes to these conditional covariance matrices.
\begin{lemma} \label{lem:ls+-covariance}
$M_{X,\alpha}$ is PSD if and only if $\Sigma_{X,\alpha}$ is PSD.
\end{lemma}
The proof is essentially identical to that of Lemma~\ref{lem:sa+-covariance}.

We note that we have not seen this hierarchy defined in this form in previous work, but it is dual to the static \LSplus proof system defined in \cite{GHP02} and described below in Section~\ref{sec:proof-systems} (see Proposition~\ref{prop:static-ls+-sdp-iff-ref}).  For a random instance $\calI$ of CSP($P$), we will study feasibility of $\mathrm{StaticLS}_+^r(R_{\calI})$.
\begin{customtheorem}{\ref{thm:static-ls+-intro}}[restated]
Let $P:[q]^k \to \{0,1\}$ be $(t-1)$-wise uniform-supporting and let $\mathcal{I}$ be a random instance of $\CSP(P)$ with $n$ variables and $m \leq n^{t/2 - \eps}$ constraints.  Then with high probability $\mathrm{StaticLS}_+^r(R_{\calI})$ is nonempty for $r=\Omega(n^{\frac{\eps}{t-2}})$ rounds.
\end{customtheorem}

\subsection{The dual point of view: Semialgebraic proof systems} \label{sec:proof-systems}
We can also consider refutation of CSPs via semialgebraic proof systems.  Starting from a set of polynomial inequalities $A = \{g_1(x) \geq 0, g_2(x) \geq 0, \ldots, g_m(x) \geq 0\}$ called axioms that encode the constraints of the CSP, semialgebraic proof systems derive new inequalities that are implied by $A$.  To prove that an instance is unsatisfiable, we wish to derive the contradiction $-1 \geq 0$.  We consider the \SA, \SAplus, \LSplus, and static \LSplus proof systems.  In this section, we again think of $\indic{\{x_T = \alpha\}}$ as a polynomial.  We give definitions for binary alphabets.  For larger alphabets, we can rewrite constraints in terms of binary variables as described above.  When refuting a CSP $\calI$, we start with constraints $A_{\calI}$ for \SA, \SAplus, and static \LSplus and use constraints $L_{\calI}$ in the case of \LSplus. 

\paragraph{The \SA proof system} An \SA refutation of $A$ has the form
\begin{equation*}
\sum_{g \in A} g(x) \sum_i \gamma_{g,i} \cdot \indic{\{x_{T_{g,i}} = \alpha_{g,i}\}}(x) + \sum_j (x_j^2-x_j) h_j(x) = -1,
\end{equation*}
where $\gamma_{g,i} \geq 0$ and the $h_j$'s are arbitrary polynomials.   This is a proof of unsatisfiability because under the assumption that the all $x_i$ variables are in $\{0,1\}$, every term of the above sum most be nonnegative and it is therefore a contradiction.  The \emph{degree} of this proof is the maximum degree of any of the terms.  The \emph{size} of the proof is the number of terms in the sum.  An degree-$r$ \SA refutation exists if and only if the corresponding $r$-round \SA relaxation is unsatisfiable; this follows from Farkas' Lemma.  The \SA proof system is automatizable: A degree-$r$ \SA refutation may be found in time $n^{O(r)}$ by solving an LP if it exists.  The \SA proof system first appeared in \cite{GHP02} with the name static $\LSop^{\infty}$; the dual hierarchy of LP relaxations was introduced by \cite{SA90}.

Lower bounds of Benabbas et al. \cite{BGMT12} and O'Donnell and Witmer \cite{OW14} immediately imply that there are no degree-$n^{\eps/(t-2)}$ \SA refutations for random instances of CSP$(P)$ with $(t-1)$-wise uniform supporting $P$ and $m \leq n^{t/2-\eps}$.
\begin{corollary}
Let $P:[q]^k \to \{0,1\}$ be $(t-1)$-wise uniform-supporting and let $\mathcal{I}$ be a random instance of $\CSP(P)$ with $n$ variables and $m \leq n^{t/2 - \eps}$ constraints.  Then with high probability there is no degree-$r$ \SA refutation of $A_{\calI}$ for $r=\Omega(n^{\frac{\eps}{t-2}})$.
\end{corollary}

\paragraph{The \SAplus proof system} In \SAplus, a proof has the form
\begin{equation*}
\sum_{f \in A} g(x) \sum_i \gamma_{g,i} \cdot \indic{\{x_{T_{g,i}} = \alpha_{g,i}\}}(x) + \sum_j (x_j^2-x_j) h_j(x) + \sum_{\ell} \eta^2_{\ell}(x) = -1
\end{equation*}
where $\gamma_{g,i} \geq 0$, the $h_j$'s are arbitrary polynomials, and the $\eta_{\ell}$'s are affine functions.  The dual \SAplus hierarchy of SDP relaxations first appeared in \cite{Rag08}.  Again, a degree-$r$ \SAplus refutation exists if and only if the corresponding $r$-round \SAplus relaxation is infeasible.  We do not know any of any results on the automatizability of \SAplus.

Our lower bound for \SAplus SDP relaxations of random instances of CSP$(P)$ implies a lower bound on the degree of \SAplus refutations of random instances of CSP$(P)$.
\begin{corollary}
Let $P:[q]^k \to \{0,1\}$ be $(t-1)$-wise uniform-supporting and let $\mathcal{I}$ be a random instance of $\CSP(P)$ with $n$ variables and $m \leq n^{t/2 - \eps}$ constraints.  Then with high probability there is no degree-$r$ \SAplus refutation of $A_{\calI}$ for $r=\Omega(n^{\frac{\eps}{t-2}})$.
\end{corollary}

\paragraph{The \LSplus proof system} The \LSplus proof system \cite{LS91} is dynamic, meaning that a proof is built up over a series of steps.  A proof in \LSplus is a sequence of polynomial inequalities $Q(x) \geq 0$.  A new inequality is derived from the inequalities already in the proof using inference rules.  When $\deg(Q(x)) \leq 1$, we allow the following:
\begin{equation*}
\infer{x_i \cdot Q(x) \geq 0}{Q(x) \geq 0} \qquad \qquad \infer{(1-x_i) \cdot Q(x) \geq 0}{Q(x) \geq 0} \qquad \qquad \infer{Q(x)^2 \geq 0}{}.
\end{equation*}
We also allow nonnegative linear combinations:
\begin{equation*}
\infer{\alpha \cdot Q(x) + \beta \cdot R(x) \geq 0}{Q(x) \geq 0 & R(x) \geq 0}
\end{equation*}
for $\alpha, \beta \geq 0$.
An \LSplus proof is therefore a sequence of ``lifting" steps in which we multiply by some $x_i$ or $(1-x_i)$ to get a degree-2 polynomial and ``projection" steps in which we take nonnegative linear combinations to reduce the degree back to $1$.  We can view an \LSplus proof as a directed acyclic graph with inequalities at each vertex and $-1 \geq 0$ at the root.  The \emph{rank} of an \LSplus proof is the maximum number of lifting steps in any path to the root.  A rank-$r$ \LSplus refutation exists if and only if the corresponding rank-$r$ \LSplus relaxation is infeasible \cite{Das01}.  The \LSplus proof system is not known to be automatizable; see Section~8 of \cite{BOG+06} for details.

Our \LSplus lower bound implies a lower bound on the rank of \LSplus refutations of random instances of CSP$(P)$.
\begin{corollary}
Let $P:[q]^k \to \{0,1\}$ be $(t-1)$-wise uniform-supporting and let $\mathcal{I}$ be a random instance of $\CSP(P)$ with $n$ variables and $m \leq n^{t/2 - \eps}$ constraints.  Then with high probability there is no rank-$r$ \LSplus refutation of $L_{\calI}$ for $r=\Omega(n^{\frac{\eps}{t-2}})$.
\end{corollary}

\paragraph{The static \LSplus proof system} A static \LSplus proof \cite{GHP02} has the following form.
\begin{equation*}
\sum_i \gamma_i \cdot b_i(x) \cdot \indic{\{x_{T_i} = \alpha_i\}}(x) + \sum_j (x_j^2-x_j) h_j(x) = -1,
\end{equation*}
where $\gamma_{i} \geq 0$, either $b_{i} \in A$ or $b_i = \eta_i^2$ for some affine function $\eta_i$, and the $h_j$'s are arbitrary polynomials.
Note that this proof system as at least as powerful as the \SAplus proof system: Terms in the sum may be products of a $\indic{\{x_{T} = \alpha\}}$ term and the square of an affine function instead of just the square of an affine function or just an axiom multiplied by a $\indic{\{x_{T} = \alpha\}}$ term.  Once again, there exists a static \LSplus refutation if and only if the corresponding static \LSplus relaxation is infeasible.  We do not know of any proof of this statement in the literature, so we include one in Appendix~\ref{sec:static-ls}.  We do not know of any results on the automatizability of static \LSplus.

Again, our static \LSplus lower bound implies a lower bound on the degree of static \LSplus refutations of CSP$(P)$.
\begin{corollary}
Let $P:[q]^k \to \{0,1\}$ be $(t-1)$-wise uniform-supporting and let $\mathcal{I}$ be a random instance of $\CSP(P)$ with $n$ variables and $m \leq n^{t/2 - \eps}$ constraints.  Then with high probability there is no degree-$r$ \LSplus refutation of $A_{\calI}$ for $r=\Omega(n^{\frac{\eps}{t-2}})$.
\end{corollary}

\subsection{Expansion} \label{sec:expansion}
Given a set of constraints $T$, we define its neighbor set $\Gamma(T) \subseteq [n]$ as $\Gamma(T) = \{v \in [n] ~|~ v \in S \text{ for some $(c,S) \in T$}\}$.
We can then define expansion.
\begin{definition}
An instance $\calI$ of \textup{CSP}$(P)$ is $(s,e)$-expanding if for every set of constraints $T$ with $|T| \leq s$, $|\Gamma(T)| \geq e|T|$.
\end{definition}

We can also define $T$'s boundary neighbors as $\partial T = \{v \in [n] ~|~ v \in S \text{ for exactly one $(c,S) \in T$}\}$
and define a corresponding notion of boundary expansion.
\begin{definition}
An instance $\calI$ of \textup{CSP}$(P)$ is $(s,e)$-boundary expanding if for every set of constraints $T$ with $|T| \leq s$, $|\partial T| \geq e|T|$.
\end{definition}

We state a well-known connection between expansion and boundary expansion appearing in, e.g., \cite{TW13}.
\begin{fact}
\label{fact:exp_to_bexp}
$(s,k-d)$-expansion implies $(s,k-2d)$-boundary expansion.
\end{fact}

It is also well-known that randomly-chosen sets of constraints have high expansion \cite{BGMT12, OW14}:
\begin{lemma}
\label{lem:exp}
Fix $\eps > 0$.  With high probability, a set of $m = n^{t/2-\eps}$ constraints chosen uniformly at random is both $\left(n^{\frac{\eps}{t-2}}, k-\frac{t}{2}+\frac{\eps}{2}\right)$-expanding and $\left(n^{\frac{\eps}{t-2}}, k-t+\eps\right)$-boundary expanding.
\end{lemma}
We give proofs of both of these statements in Appendix~\ref{sec:expansion-proofs}.

\subsection{Constructing consistent local distributions} \label{sec:cons-local-dists}
Here, we recall a construction of consistent local distributions supported on satisfying assignments.  We will study these distributions in the remainder of this paper, showing that they are valid \SAplus, \LSplus, and static \LSplus solutions.  They were first used in \cite{BGMT12} and have appeared in many subsequent works (e.g, \cite{TW13, OW14, BCK15}).  In Appendix~\ref{sec:cons-local-proofs}, we give proofs of all results mentioned in this section.

We first need to mention the notion of a closure of a set of variables.  For $S \subseteq [n]$, let $H_{\calI}-S$ denote the hypergraph $H_{\calI}$ with the vertices of $S$ and all hyperedges contained in $S$ removed.  Intuitively, the closure $\cl(S)$ of a set $S \subseteq [n]$ is a superset of $S$ that is not too much larger than $S$ and is not very well-connected to the rest of the instance in the sense that $H_{\calI} - S$ has high expansion.
\begin{lemma} [\textup{\cite{BGMT12, TW13}}] \label{lem:closure}
If $H_{\calI}$ is $(s_1,e_1)$-expanding and $S$ is a set of variables such that $|S| < (e_1-e_2)s_1$ for some $e_2 \in (0,e_1)$, then there exists a set $\cl(S) \subseteq [n]$ such that $S \subseteq \cl(S)$ and $H_{\calI} - \cl(S)$ is $(s_2,e_2)$-expanding with $s_2 \geq s_1-\frac{|S|}{e_1-e_2}$ and $\cl(S) \leq \frac{k+2e_1-e_2}{2(e_1-e_2)}|S|$.
\end{lemma}
We give a formal definition of the closure and a proof of this lemma in Appendix~\ref{sec:cons-local-proofs}.

We now use the closure to define consistent local distributions supporting on satisfying assignments.  We assume that there exists a $(t-1)$-wise independent distribution $\mu$ over satisfying assignments to $P$.  For a constraint $C = (c,S)$, let $\mu_C$ be the distribution defined by $\mu_C(z) = \mu(z_1+c_1,\ldots,z_k+c_k)$ and let $\calC(S)$ be the set of constraints whose support is entirely contained within $S$.  For a set of variables $S \subseteq [n]$ and an assignment $\alpha \in [q]^S$ , we use the notation $S = \alpha$ to indicate the the variables of $S$ are labeled according to the assignment $\alpha$.  For a constraint $C=(c,S)$ and an assignment $\alpha$ to a superset of $S$, let $\mu_C(\alpha) = \mu_C(\alpha_S)$.

For $S \subseteq [n]$, we can then define the distribution $\Pi'_S$ over $[q]^S$ as
\[
\Pi'_S(S = \alpha) =  \frac{1}{Z_{S}}  \prod_{C \in \calC(S)} \mu_C(\alpha), \text{ where } Z_{S} = \sum_{\beta \in [q]^{S}} \prod_{C \in \calC(S)} \mu_C(\beta).
\]
Using $\{\Pi'_S\}$, we can then define local distributions $\Pi_S$ by $\Pi_S(S = \alpha) = \Pi'_{\cl(S)}(S = \alpha)$.
\cite{BGMT12, OW14} proved that these distributions are $r$-locally consistent for $r = n^{\frac{\eps}{t-2}}$.
\begin{theorem}[\textup{\cite{BGMT12, OW14}}] \label{thm:local-cons}
For a random instance $\calI$ with $m \leq n^{t/2-\eps}$, the family of distributions $\{\Pi_S\}_{|S| \leq r}$ is $r$-locally consistent for $r = \Omega\left(n^{\frac{\eps}{t-2}}\right)$ and is supported on satisfying assignments.
\end{theorem}

This theorem shows that the \SA cannot efficiently refute random $(t-1)$-wise uniforming supporting instances: the $r$-round \SA LP still has value $1$ for some $r = \Omega(n^{\frac{\eps}{t-2}})$ when $m \leq n^{t/2-\eps}$.  In this paper, we show that even when we add the \SAplus requirement that the covariance matrix is PSD, we still cannot refute when $m \leq n^{t/2-\eps}$.

Given these locally consistent distributions, Lemma~\ref{lem:cond-local-cons} implies that the conditional distributions $\{\Pi_{S|X=\alpha}\}$ defined above are also locally consistent.
\begin{corollary}[\textup{\cite{TW13}}]\label{cor:cond-local-cons}
There exists a constant $c > 0$ such that the following holds.  Let $\calI$ be a random instance of $\CSP(P)$ with $m \leq n^{t/2-\eps}$.  Let $X \subseteq [n]$ such that $|X| \leq c n^{\frac{\eps}{t-2}}$ and let $\alpha \in \{0,1\}^X$ be any assignment to $X$ such that $\mu_C(\alpha) > 0$ for all constraints in $C \in \calC(X)$.  Then the family of conditional distributions $\{\Pi_{S| X = \alpha}\}_{|S| \leq r, S \cap X = \emptyset}$ is $r$-locally consistent for some $r = \Omega(n^{\frac{\eps}{t-2}})$.
\end{corollary}
We will use these conditional consistent local distributions to prove lower bounds for \LSplus in Section~\ref{sec:ls}.

\section{Overview of the proof}
Showing that a set of local distributions is a valid \SAplus solution requires proving that these distributions are locally consistent and proving that their covariance matrix is PSD.  Local consistency of the $\{\Pi_S\}$ distributions was proven in previous work \cite{BGMT12, OW14}.  To prove Theorem~\ref{thm:sa+-intro}, it remains to argue that the covariance matrix of $\{\Pi_S\}$ is PSD.

Previous work \cite{BGMT12, TW13} only considers instances with a linear number of constraints and relies on the fact that most pairs of variables are uncorrelated in this regime.  For $m \gg n$, however, correlations between pairs of vertices do arise because the underlying hypergraph becomes more dense.  The major technical contribution of this work is to deal with these correlations by proving that they remain local.  We consider the graph induced by correlations between variables: Two variables are connected if they have non-zero correlation.  We prove that this graph must have connected components of at most constant size with high probability.  Each of these connected components can then be covered by a local distribution of constant size.  This implies that each submatrix of the covariance matrix corresponding to one of these connected component is PSD, and thus the entire covariance matrix is PSD.

The proof of Theorem~\ref{thm:sa+-intro} has three steps.  First, we show in Section~\ref{sec:corr-graph} that if the correlation graph has small connected components, then the covariance matrix is PSD.  Second, we show that any non-zero correlation must have been caused by a relatively dense subset of constraints in Section~\ref{sec:bad-structures}.  In Section~\ref{sec:small-conn-comps}, we show that connected components in the correlation graph must be small or they would induce large dense subsets of constraints that would violate expansion properties.

In Section~\ref{sec:ls}, we show that this same strategy can be used to prove PSDness of conditional covariance matrices and thereby prove Theorems~\ref{thm:ls+-intro} and \ref{thm:static-ls+-intro}.

\section{The correlation graph}
\label{sec:corr-graph}
In this section, we define the correlation graph and show that if the correlation graph has small connected components, then the covariance matrix is PSD.\@

\if0
\begin{definition}
The covariance matrix $\Sigma$ is defined as
\[
\Sigma_{(u,a),(v,b)} = \Pr_{D(\{u,v\})}[u = a \wedge v = b] - \Pr_{D(\{u\})}[u = a] \cdot \Pr_{D(\{v\})}[v = b] 
\]
\end{definition}

\begin{lemma} \label{lem:covariance}
There are \SAplus vectors if and only if the covariance matrix is PSD.
\end{lemma}
\begin{proof}
Lemma~3.2 of \cite{BGMT12} implies that we can get \SAplus vectors if the matrix
\[
M = \begin{pmatrix}
1 & w^T \\
w & B
\end{pmatrix}
\]
is PSD, where $w_{(u,a)} = \Pr_{D(\{u\})}[u = a]$ and $B_{(u,a),(v,b)} = \Pr_{D(\{u,v\})}[u = a \wedge v = b]$.

Applying the Schur complement condition for PSD-ness (Is this well-known?  I found it on Wikipedia.), we see that $M$ is PSD if and only if
\[
B - ww^T = \Sigma
\]
is PSD.
\end{proof}
\fi

\begin{definition}
The correlation graph $G_{\mathrm{corr}}$ associated with $r$-locally consistent distributions $\{D_S\}$ 
is the graph on $[n]$ with an edge between every pair of variables for which there is a nonzero entry in the covariance matrix for $\{D_S\}$.  More formally, the set of edges of $G_{\mathrm{corr}}$ is defined to be
\[
E(G_{\mathrm{corr}}) = \{(u,v) \in [n] \times [n] ~|~ u\ne v, \exists (a,b) \in [q] \times [q] \text{ s.t. } \Sigma_{(u,a),(v,b)} \ne 0\}.
\]
\end{definition}

\begin{lemma} \label{lem:small-conn-comps-implies-psd}
Let $\{D_S\}$ be a family of $r$-locally consistent distributions. 
If all connected components in the correlation graph associated with $\{D_S\}$ have size at most $r$, then the covariance matrix for $\{D_S\}$ is PSD.\@
\end{lemma}
\begin{proof}
Consider the partition $V_1,V_2,\ldots,V_{\ell}$ of $[n]$ such that $u$ and $v$ are in the same set if and only if they are connected in the correlation graph.  We then have nonzero entries in the covariance matrix only for pairs $((u,a),(v,b))$ such that $u,v \in V_i$ for some $i$.  
Ordering the rows and columns of the covariance matrix according to this partition, we see that the covariance matrix is block diagonal with a nonzero block for each connected component of the correlation graph.
Each of these blocks is PSD since each is the covariance matrix of the local distribution $D_{V_i}$ for the corresponding set $V_i$ with size at most $r$, and the covariance matrix of valid distribution is always PSD.\@  Since each block is PSD, the entire matrix is PSD.\@
\end{proof}
We already know that $\{\Pi_S\}$ defined in Section~\ref{sec:cons-local-dists} is $\Omega(n^{\frac{\eps}{t-2}})$-locally consistent with high probability when $m \leq n^{t/2-\epsilon}$.
In the following sections, we will show that connected components in the correlation graph associated with $\{\Pi_S\}$ are small.
Hence, from Lemma~\ref{lem:small-conn-comps-implies-psd}, $\{\Pi_S\}$ is a feasible solution for the \SAplus SDP.

\section{Correlations are induced by small, dense structures} \label{sec:bad-structures}
In this section, we show that pairwise correlations in $\{\Pi_S\}$ are only generated by small, dense subhypergraphs that we will call ``bad structures".
Given a set of hyperedges $W$, call a variable $v$ an $W$-boundary variable if it is contained in exactly one constraint in $W$.

\begin{definition}
\label{def:bad_struct}
For variables $u$ and $v$, a bad structure for $u$ and $v$ is a set of constraints $W$ satisfying the following properties:
\if0
\begin{enumerate}
\item
\label{def:complete}
There exist $C_u,C_v \in Y$ such that $u \in C_u$ and $v \in C_v$.
\item
\label{def:be_bound}
$|\partial Y| \leq (k-t)|Y| + 2$.
\item
\label{def:one_bv}
At least one of $u$ and $v$ is an $Y$-boundary variable.
\item
\label{def:no_var}
Every constraint not containing $u$ or $v$ as one of its $Y$-boundary variables contains at most $k-t$ $S$-boundary variables.
\item
\label{def:one_var}
Any constraint containing one of $u$ or $v$ as an $S$-boundary variable has at most $k-t+1$ $S$-boundary variables.
\item
\label{def:two_vars}
Any constraint containing both of $u$ and $v$ as $S$-boundary variables has at most $k-t+2$ $S$ boundary variables.
\item 
\label{def:conn}
There exist no two partitions $S_1,S_2$ of $S$ and $R_1,R_2$ of $\Gamma(S)$ such that $\Gamma(S_1) = R_1$ and $\Gamma(S_2)= R_2$. ($S$ is connected.)
\end{enumerate}
\fi
\begin{enumerate}
\item
\label{def:complete}
$u, v\in\Gamma(W)$.
\item 
\label{def:conn}
The hypergraph induced by $W$ is connected.
\item
\label{def:boundary}
Every constraint contains at most $k-t$ $W$-boundary variables other than $u$ and $v$.
\end{enumerate}
We also say $W$ is a bad structure if $W$ is a bad structure for some $u$ and $v$.
\end{definition}
A bad structure for $u$ and $v$ generates correlation between $u$ and $v$ with respect to $\{\Pi_S\}$.
\begin{lemma}
\label{lem:bad_struct}
If there is no bad structure for $u$ and $v$ of size at most $|\calC(\cl(\{u,v\}))|$, then $u$ and $v$ are not correlated with respect to $\Pi_{\{u,v\}}$.
\end{lemma}

We need the following technical claim, which states that the distribution $\Pi'_S$ isn't affected by removing a constraint with many boundary variables.
\begin{claim} \label{cl:rem-one-cons}
Let $T \subseteq S \subseteq [n]$ be sets of variables.  Let $C^* \in \calC(S)$ be some constraint covered by $S$.  If $|(\partial \calC(S) \cap C^*) \setminus T| \geq k-t+1$, then for any $\alpha \in \{0,1\}^T$,
\[
\Pi'_S(T = \alpha) \propto \Pi'_{S \setminus (\partial \calC(S) \cap C^*)}(T \setminus (\partial \calC(S) \cap C^*) = \alpha_{T \setminus (\partial \calC(S) \cap C^*)}),
\]
\end{claim}
\begin{proof}
Let $B = \partial \calC(S) \cap C^*$ be the boundary variables of $\calC(S)$ contributed by $C^*$, i.e., the variables contained in $C^*$ that don't appear in any other constraint of $\calC(S)$.  Then
\begin{align*}
\Pi'_S(T = \alpha) &\propto \sum_{\substack{\beta \in \{0,1\}^S \\ \beta_T = \alpha}} \prod_{C \in \calC(S)} \mu_C(\beta) \\
&= \sum_{\substack{\beta \in \{0,1\}^{S \setminus B} \\ \beta_{T \setminus B} = \alpha_{T \setminus B}}} \prod_{C \in \calC(S) \setminus \{C^*\}} \mu_C(\beta) \sum_{\substack{\gamma \in \{0,1\}^B \\ \gamma_{B \cap T} = \alpha_{B \cap T}}} \mu_{C^*}(\beta, \gamma) \\
&= \frac{1}{q^{k-|B \setminus T|}}\sum_{\substack{\beta \in \{0,1\}^{S \setminus B} \\ \beta_{T \setminus B} = \alpha_{T \setminus B}}} \prod_{C \in \calC(S) \setminus \{C^*\}} \mu_C(\beta) \\
&\propto \Pi'_{S \setminus B}(T \setminus B = \alpha_{T \setminus B}).
\end{align*}
The second-to-last line holds because $|B \setminus T| \geq k-t+1$ and $\mu$ is $(t-1)$-wise independent.
\end{proof}

Using Claim~\ref{cl:rem-one-cons}, we prove Lemma~\ref{lem:bad_struct}.
\begin{proof}[Proof of Lemma~\ref{lem:bad_struct}]
Let $S_0 = \calC(\cl(\{u,v\}))$.  Say there exists a constraint $C_1$ such that $|(\partial \calC(S_0) \cap C_1) \setminus \{u,v\}| \geq k-t+1$.  Let $S_1 = S_0 \setminus (\partial \calC(S_{0}) \cap C_1)$.  If there exists a constraint $C_2$ such that $|(\partial \calC(S_1) \cap C_2) \setminus \{u,v\}| \geq k-t+1$, remove its boundary variables in the same manner to get $S_2$.  Continue in this way until we obtain a set $S_{\ell}$ such that $|(\partial \calC(S_{\ell}) \cap C) \setminus \{u,v\}| \leq k-t$ for every constraint $C \in \calC(S_{\ell})$ ($\calC(S_{\ell})$ could be empty).  Since $|(\partial \calC(S_{i-1}) \cap C_i) \setminus \{u,v\}| \geq k-t+1$ for $1 \leq i \leq \ell$, we can apply Claim~\ref{cl:rem-one-cons} $\ell$ times to see that
\[
\Pi_{\{u,v\}}(u = a \wedge v = b) \propto \begin{cases}
\Pi'_{S_{\ell}}(u = a) &  \text{if $u \in S_{\ell}, v \notin S_{\ell}$} \\
\Pi'_{S_{\ell}}(v = b) &  \text{if $v \in S_{\ell}, u \notin S_{\ell}$} \\
1 & \text{if $u,v \notin S_{\ell}$} \\
\Pi'_{S_{\ell}}(u = a \wedge v = b) & \text{if $u,v \in S_{\ell}$}. \\
\end{cases}
\]
In the first three cases, it is easy to see that the lemma holds.  In the last case, we know that $\calC(S_{\ell})$ cannot be a bad structure by our assumption.  Since $\calC(S_{\ell})$ satisfies Conditions~\ref{def:complete} and \ref{def:boundary} of Definition~\ref{def:bad_struct}, the hypergraph induced by $S_{\ell}$ must be disconnected with $u$ and $v$ in different connected components.  Say $S_u$ and $S_v$ are the vertex sets of the connected components of $S_{\ell}$ containing $u$ and $v$, respectively.  Then
\[
\Pi_{\{u,v\}}(u = a \wedge v = b) \propto \Pi'_{S_{\ell}}(u = a \wedge v = b) = \Pi'_{S_u}(u = a) \cdot \Pi'_{S_v}(v = b).
\]
The result then follows.
\end{proof}

\section{All connected components of the correlation graph are small} \label{sec:small-conn-comps}
In this section, we show that all connected components in the correlation graph associated with $\{\Pi_S\}$ are small, which concludes the proof of Theorem~\ref{thm:sa+-intro}.
\begin{theorem} \label{thm:small-conn-comps-new}
Assume that the hypergraph $H_{\calI}$ is an $(r, k-t/2+\delta/2)$-expander for some $r = \omega(1)$.
Then all connected components in the correlation graph associated with $\{\Pi_S\}$ have size at most $\frac{2k}{\delta}$.
\end{theorem}
We will actually prove a slightly more general theorem that we will use to prove \LSplus lower bounds in Section~\ref{sec:ls}.  Given a hypergraph $H$, let $G_{\mathrm{bad}}(H)$ be the graph on $[n]$ such that there is an edge between $i$ and $j$ if and only if there exists a bad structure for $i$ and $j$ in $H$.
\begin{theorem} \label{thm:small-conn-comps-gb}
If the hypergraph $H$ is an $(r, k-t/2+\delta/2)$-expander for some $r = \omega(1)$, then all connected components in $G_{\mathrm{bad}}(H)$ have size at most $\frac{2k}{\delta}$.
\end{theorem}
Lemma~\ref{lem:bad_struct} implies that $G_{\mathrm{bad}}(H_{\calI})$ contains the correlation graph associated with $\{\Pi_S\}$ as a subgraph, so Theorem~\ref{thm:small-conn-comps-gb} immediately implies Theorem~\ref{thm:small-conn-comps-new}.
\begin{proof}[Proof of Theorem~\ref{thm:small-conn-comps-gb}]
For any edge $e$ of $G_{\mathrm{bad}}(H)$, we can find a corresponding bad structure $W_e$.
We will say that $W_e$ induces $e$.
Any such bad structure $W_e$ satisfies
\begin{equation}
\label{eq:size-bad}
\Gamma(W_e) \le (k-t)|W_e| + 2 + \frac{k|W_e| - ((k-t)|W_e|+2)}2
= \left(k-\frac{t}{2}\right)|W_e| + 1.
\end{equation}
The first term upper bounds the number of boundary vertices, the second term counts the endpoints of $e$, and the last term upper bounds the number of non-boundary vertices.  For a connected component in $G_{\mathrm{bad}}(H)$, let $e_1, e_2, \dotsc, e_\ell$ be an ordering of edges in the connected component such that $(\bigcup_{j=1}^i e_j)\cap e_{i+1}$ is not empty for $i=1,\dotsc,\ell$.  That is, $e_1, e_2, \dotsc, e_\ell$ is an ordering of the edges in the connected component such that every edge except for the first one is adjacent to some edge preceding it.  Let $W_{e_1},\dotsc, W_{e_\ell}$ be corresponding bad structures inducing these edges.
Let $T_i = \bigcup_{j=1}^i W_{e_j}$ for $i=1,\dotsc,\ell$.
While $T_i$ itself is not necessarily a bad structure, we will show that the inequality~\eqref{eq:size-bad} still holds for $T_i$, i.e.,
\begin{equation}
\label{eq:size-bad-union}
\Gamma(T_i) \le \left(k-\frac{t}{2}\right)|T_i| + 1
\end{equation}
for any $i=1,\dotsc,\ell$.
If~\eqref{eq:size-bad-union} holds, the number of constraints in $T_\ell$ is at most $\frac{2}{\delta}$; otherwise, expansion is violated.
Hence, at most $\frac{2k}{\delta}$ vertices are included in the connected component of the correlation graph associated with $\{D_S\}$.

In the following, we prove~\eqref{eq:size-bad-union}.  First, note that $|\Gamma(T_1)| \leq \left(k-\frac{t}{2}\right)|T_1| + 1$ by \eqref{eq:size-bad}.
Let $W_i' = W_{e_i}\setminus T_{i-1}$ be the new constraints added at step $i$.  Call any vertex in $\Gamma(T_i)\setminus \Gamma(T_{i-1})$ a new vertex. We will prove that at most $\left(k-\frac{t}{2}\right)|W_i'|$ new vertices are added and this will imply~\eqref{eq:size-bad-union}.

Let $n_i$ be the number of new $W_i'$-boundary vertices.
Then the total number of new vertices is at most
\begin{equation} \label{eq:num-new-bdry}
n_i + (k|W_i'|-1 - n_i) / 2.
\end{equation}
The second term upper bounds the number of non-boundary vertices.   The $-1$ comes from the fact that $\Gamma(W_i')$ must intersect $\Gamma(T_{i-1})$ since $e_i$ must be adjacent to some preceding edge.  If $\Gamma(W_i')$ and $\Gamma(T_{i-1})$ intersect in a boundary vertex, the resulting bound is stronger.

Hence, we would like to upper bound $n_i$.
We know that $n_i$ is at most $(k-t)|W_i'|+1$ since
any new $W_i'$-boundary vertex must be a new $W_{e_i}$-boundary vertex, all but one constraint in $W_i'$
have at most $k-t$ new $W_{e_i}$-boundary vertices, and one constraint in $W_i'$ has at most $k-t+1$ new $W_{e_i}$-boundary vertices.
Plugging this into \eqref{eq:num-new-bdry}, we see that the number of new vertices is at most
$(k-t)|W_i'| + 1 + (k|W_i'| - 1 - ((k-t)|W_i'| + 1))/2 = (k-t/2)|W_i'|$.
\end{proof}

From Lemmas~\ref{lem:exp} and \ref{lem:small-conn-comps-implies-psd} and Theorems~\ref{thm:local-cons} and~\ref{thm:small-conn-comps-new},
we obtain Theorem~\ref{thm:sa+-intro}.

\section{\LSplus rank lower bounds} \label{sec:ls}
In this section, we use techniques from the previous sections to prove PSDness of the moment matrices $M_{X,\alpha}$ of the conditional local distributions $\{\Pi_{S|X=\alpha}\}$.  From here, degree lower bounds for the static \LSplus proof system and rank lower bounds for \LSplus follow easily.

\begin{lemma} \label{lem:psd-protection}
There exists a constant $c > 0$ such that the following holds.  Let $X \subseteq [n]$ such that $|X| \leq cn^{\frac{\delta}{t-2}}$.
For any $\alpha \in [q]^{X}$ such that $\mu_C(\alpha) > 0$ for all $C \in \calC(X)$, $M_{X,\alpha}$ is positive semidefinite.
\end{lemma}
Since we already know that $\{\Pi_S\}$ is a valid \SA solution for $|S| = \Omega(n^{\frac{\delta}{t-2}})$, this lemma immediately implies Theorem~\ref{thm:static-ls+-intro}.  Theorem~\ref{thm:ls+-intro}, our rank lower bound for \LSplus refutations, follows from Theorem~\ref{thm:static-ls+-intro} and the following fact.

\begin{fact}
If there exists a rank-$r$ \LSplus refutation of a set of axioms $A$, then there exists a static \LSplus refutation of $A$ with degree at most $r$.
\end{fact}
\begin{proof}
Let $R$ be a rank-$r$ \LSplus refutation.  We look at $R$ as a directed acyclic graph in which each node is the application of some inference rule, the root is $-1 \geq 0$, and the leaves are axioms or applications of the rule $h(x)^2 \geq 0$ for some $h$ with degree at most $1$.  
Starting from the leaves and working back to the root $-1 \geq 0$, we can substitute in the premises of each inference to get an expression $Q(x) = -1$.
Since $R$ has rank $r$, each path in $r$ has at most $r$ multiplications by a term of the form $x_i$ or $(1-x_i)$ and $Q(x) = -1$ must be a valid static \LSplus refutation of degree at most $r$. 
\end{proof}

To prove Lemma~\ref{lem:psd-protection}, we first show that $M_{X,\alpha}$ is PSD when $H-X$ has high expansion.
Then we show that any $M_{X,\alpha}$ can expressed as a nonnegative combination of $M_{\cl(X),\beta}$'s for $\beta \in [q]^{\cl(X)}$.
Since $H - \cl(X)$ has high expansion when $|X| \leq cn^{\frac{\delta}{t-2}}$, each of the $M_{\cl(X),\beta}$'s is PSD.  $M_{X,\alpha}$ is therefore a nonnegative combination of PSD matrices and must itself be PSD.

We start by generalizing Lemma~\ref{lem:bad_struct} to conditional distributions.
\begin{lemma}
\label{lem:bad_struct_cond}
Let $X \subseteq [n]$ and $\alpha \in [q]^X$ such that $\mu_C(\alpha) > 0$ for all $C \in \calC(X)$.
If there is no bad structure for $u$ and $v$ in $H-X$ of size at most $|\calC(\cl(\{u,v\}) \setminus X)|$, then $u$ and $v$ are not correlated with respect to $\Pi_{\{u,v\}|X = \alpha}$.
\end{lemma}
\begin{proof}
First, recall that
\[
\Pi_{\{u,v\}|X=\alpha}(u=a \wedge v=b) = \frac{\Pi_{\{u,v\} \cup X}(u = a \wedge v = b \wedge X = \alpha)}{\Pi_{X}(X = \alpha)}.
\]
We will show that $\Pi_{\{u,v\} \cup X}(u = a \wedge v = b \wedge X = \alpha)$ is equal to the product of a term depending on $u$ and $a$ but not $v$ and $b$ and a term depending on $v$ and $b$ but not $u$ and $a$.  From there, the lemma immediately follows.

The proof is essentially the same as that of Lemma~\ref{lem:bad_struct} above.   Starting with $S_0 = \calC(\cl(\{u,v\}))$, we apply the same process except we require that each constraint $C_i$ that we remove satisfies $|(\partial \calC(S_{i-1}) \cap C_i) \setminus (\{u,v\} \cup X)| \geq k-t+1$.  At the end of this process, we are left with a set $S_{\ell}$ such that $|(\partial \calC(S_{\ell}) \cap C) \setminus (\{u,v\} \cup X)| \leq k-t$ for every constraint $C \in \calC(S_{\ell})$ (again, $\calC(S_{\ell})$ could be empty).  Let $X_{\ell} = X \cap \Gamma(S_{\ell})$ and let $\alpha_{\ell} = \alpha_{X_{\ell}}$.  By applying Lemma~\ref{cl:rem-one-cons} repeatedly, we see that
\[
\Pi_{\{u,v\} \cup X}(u = a \wedge v = b \wedge X = \alpha) \propto \begin{cases}
& \Pi'_{S_{\ell}}(u=a \wedge X_{\ell} = \alpha_{\ell}) \qquad \text{if $u \in S_{\ell}, v \notin S_{\ell}$}  \\
& \Pi'_{S_{\ell}}(v=b \wedge X_{\ell} = \alpha_{\ell}) \qquad \text{if $v \in S_{\ell}, u \notin S_{\ell}$} \\
& \Pi'_{S_{\ell}}(X_{\ell} = \alpha_{\ell}) \qquad \text{if $u,v \notin S_{\ell}$} \\
& \Pi'_{S_{\ell}}(u=a \wedge v=b \wedge X_{\ell} = \alpha_{\ell}) \qquad \text{if $u,v \in S_{\ell}$}.
\end{cases}
\]
In all cases except for the last one, the result follows.  In the last case, the assumption that there is no bad structure in $H - X$ implies that $H[S_{\ell} \setminus X]$ must be disconnected with $u$ and $v$ in separate connected components just as in the proof of Lemma~\ref{lem:bad_struct}.  If $u$ and $v$ are also in separate connected components in $H[S_{\ell}]$, then it is easy to see that the lemma holds.

Otherwise, $u$ and $v$ are in the same connected component in $H[S_{\ell}]$; we denote its edges by $E$.  Since $u$ and $v$ are in separate connected components of $H[S_{\ell} \setminus X]$, we know that $E \setminus \calC(X)$ has separate connected components with edge sets $E_u$ and $E_v$ containing $u$ and $v$, respectively.  Let $S_u = \Gamma(E_u)$, $S_v = \Gamma(E_v)$, and $S_{\mathrm{rest}} = S_{\ell} \setminus (S_u \cup S_v)$.  Let $X_u = X \cap S_u$ and $\alpha_u = \alpha_{X_u}$.  Define $X_v$, $X_{\mathrm{rest}}$, $\alpha_v$, and $\alpha_{\mathrm{rest}}$ in the same way.  We can then write $\Pi'_{S_{\ell}}(u=a \wedge v=b \wedge X_{\ell} = \alpha_{\ell})$ as
\[
 \Pi'_{S_u}(u=a \wedge X_u = \alpha_u) \cdot \Pi'_{S_v}(v=b \wedge X_v = \alpha_v) \cdot \Pi'_{S_{\mathrm{rest}}}(X_{\mathrm{rest}} = \alpha_{\mathrm{rest}}).
\]
Since $\Pi'_{S_{\mathrm{rest}}}(X_{\mathrm{rest}} = \alpha_{\mathrm{rest}})$ depends only on $\alpha$, the lemma follows.
\end{proof}

Using this lemma, we can prove that $M_{X,\alpha}$ is PSD when $H-X$ has high enough expansion.
\begin{lemma} \label{lem:psd-protection-high-exp}
Let $X \subseteq [n]$ such that $H - X$ is $\left(r, k-t/2+\eps\right)$-expanding for $r = \omega(1)$ and some constant $\eps > 0$.  Then for any $\alpha \in \{0,1\}^{X}$ with $\mu_C(\alpha) > 0$ for all $C \in \calC(X)$, $M_{X,\alpha}$ is positive semidefinite.
\end{lemma}
\begin{proof}
By Lemma~\ref{lem:ls+-covariance}, $M_{X,\alpha}$ is PSD if and only if $\Sigma_{X,\alpha}$ is PSD, so it suffices to show that $\Sigma_{X,\alpha}$ is PSD.
The conditional distributions $\{\Pi_{S|X = \alpha}\}$ are $r$-locally consistent for $r= \Omega(n^{\frac{\delta}{t-2}})$ by Corollary~\ref{cor:cond-local-cons}.   Then Lemma~\ref{lem:small-conn-comps-implies-psd} implies that $\Sigma_{X,\alpha}$ is PSD if the correlation graph of the $\{\Pi_{S|X = \alpha}\}$ distributions has connected components of size at most $r$.   Lemma~\ref{lem:bad_struct_cond} implies that correlations under $\{\Pi_{S|X = \alpha}\}$ induce bad structures in $H-X$, and we can apply Theorem~\ref{thm:small-conn-comps-gb} to $G_{\mathrm{bad}}(H-X)$ to complete the proof.
\end{proof}

Finally, we show that for any $X$, $M_{X,\alpha}$ can be expressed as a nonnegative combination of $M_{\cl(X),\beta}$'s for
$\beta \in [q]^{\cl(X)}$.
As Lemma~\ref{lem:psd-protection-high-exp} implies that each $M_{\cl(X),\beta}$ is PSD, $M_{X,\alpha}$ is a nonnegative combination of PSD matrices.  This implies that $M_{X,\alpha}$ is PSD for any small enough $X$ and any 
$\alpha \in [q]^{X}$
with $\mu_C(\alpha) > 0$ for all $C \in \calC(X)$, completing the proof of Lemma~\ref{lem:psd-protection}.

\begin{claim} \label{cl:psd-sum}
Assume that $\{D_S\}$ is a family of $r$-locally consistent distributions.  Then
\[
M_{X,\alpha} = \sum_{\substack{\beta\in [q]^T \\ \beta_X=\alpha}} D_{T|X=\alpha}(T = \beta) \cdot M_{T,\beta}
\]
for any $X\subseteq T$ such that $|T| \leq r$.
\end{claim}
The proof of this claim is immediate from the definitions of $D_{T|X=\alpha}$ and $M_{X,\alpha}$.

\section*{Acknowledgments}
The first-named author would like to thank Osamu Watanabe for his encouragement.
The second-named author would like to thank Anupam Gupta and Ryan O'Donnell for several helpful discussions.  We would also like to thank several anonymous reviewers for constructive comments.

\bibliographystyle{alpha}
\bibliography{/Users/dwitmer/Dropbox/Latex/witmer}

\newcommand{\etalchar}[1]{$^{#1}$}
\begin{thebibliography}{BOGH{\etalchar{+}}06}

\bibitem[ABW10]{ABW10}
Benny Applebaum, Boaz Barak, and Avi Wigderson.
\newblock Public-key cryptography from different assumptions.
\newblock In {\em Proceedings of the 42nd ACM Symposium on Theory of
  Computing}, pages 171--180, 2010.

\bibitem[AOW15]{AOW15}
Sarah~R. Allen, Ryan O'Donnell, and David Witmer.
\newblock How to refute a random {CSP}.
\newblock In {\em Proceedings of the 56th Annual IEEE Symposium on Foundations
  of Computer Science}, pages 689--708, 2015.

\bibitem[BCK15]{BCK15}
Boaz Barak, Siu~On Chan, and Pravesh Kothari.
\newblock Sum of squares lower bounds from pairwise independence.
\newblock In {\em Proceedings of the 47th Annual ACM Symposium on Theory of
  Computing}, pages 97--106, 2015.

\bibitem[BGMT12]{BGMT12}
Siavosh Benabbas, Konstantinos Georgiou, Avner Magen, and Madhur Tulsiani.
\newblock {SDP} gaps from pairwise independence.
\newblock {\em Theory of Computing}, 8:269--289, 2012.

\bibitem[BKS13]{BKS13}
Boaz Barak, Guy Kindler, and David Steurer.
\newblock On the {O}ptimality of {S}emidefinite {R}elaxations for
  {A}verage-{C}ase and {G}eneralized {C}onstraint {S}atisfaction.
\newblock In {\em Proceedings of the 4th Innovations in Theoretical Computer
  Science conference}, 2013.

\bibitem[BM15]{BM15}
Boaz Barak and Ankur Moitra.
\newblock Tensor {P}rediction, {R}ademacher {C}omplexity and {R}andom 3-{XOR}.
\newblock {\em CoRR}, abs/1501.06521, 2015.

\bibitem[BOGH{\etalchar{+}}06]{BOG+06}
Joshua Buresh-Oppenheim, Nicola Galesi, Shlomo Hoory, Avner Magen, and Toniann
  Pitassi.
\newblock Rank bounds and integrality gaps for cutting planes procedures.
\newblock {\em Theory Comput.}, 2:65--90, 2006.

\bibitem[BSB02]{BB02}
Eli Ben-Sasson and Yonatan Bilu.
\newblock A gap in average proof complexity.
\newblock {\em Electronic Colloquium on Computational Complexity {(ECCC)}},
  9(3), 2002.

\bibitem[CLP02]{CLP02}
A~Crisanti, L~Leuzzi, and G~Parisi.
\newblock The 3-{SAT} problem with large number of clauses in the
  $\infty$-replica symmetry breaking scheme.
\newblock {\em Journal of Physics A: Mathematical and General}, 35(3):481,
  2002.

\bibitem[CLRS13]{CLRS13}
Siu~On Chan, James~R. Lee, Prasad Raghavendra, and David Steurer.
\newblock Approximate constraint satisfaction requires large {LP} relaxations.
\newblock In {\em Proceedings of the 54th Annual IEEE Symposium on Foundations
  of Computer Science}, pages 350--359, 2013.

\bibitem[COGL04]{COGL04}
Amin Coja-Oghlan, Andreas Goerdt, and Andr{\'e} Lanka.
\newblock {S}trong {R}efutation {H}euristics for {R}andom $k$-{SAT}.
\newblock In Klaus Jansen, Sanjeev Khanna, Jos\'{e}~D.P. Rolim, and Dana Ron,
  editors, {\em {A}pproximation, {R}andomization, and {C}ombinatorial
  {O}ptimization. {A}lgorithms and {T}echniques}, volume 3122 of {\em {L}ecture
  {N}otes in {C}omputer {S}cience}, pages 310--321. Springer Berlin Heidelberg,
  2004.

\bibitem[Das01]{Das01}
Sanjeeb Dash.
\newblock {\em On the {M}atrix {C}uts of {L}ov{\'a}sz and {S}chrijver and their
  use in {I}nteger {P}rogramming}.
\newblock PhD thesis, Rice University, 2001.

\bibitem[DLSS14]{DLS14}
Amit Daniely, Nati Linial, and Shai Shalev-Shwartz.
\newblock From average case complexity to improper learning complexity.
\newblock In {\em Proceedings of the 46th Annual ACM Symposium on Theory of
  Computing}, pages 441--448. ACM, 2014.

\bibitem[DSS15]{DSS15}
Jian Ding, Allan Sly, and Nike Sun.
\newblock Proof of the satisfiability conjecture for large $k$.
\newblock In {\em Proceedings of the 47th Annual ACM Symposium on Theory of
  Computing}, pages 59--68, 2015.

\bibitem[Fei02]{Fei02}
Uriel Feige.
\newblock Relations {B}etween {A}verage {C}ase {C}omplexity and {A}pproximation
  {C}omplexity.
\newblock In {\em Proceedings of the 34th Annual ACM Symposium on Theory of
  Computing}, pages 534--543, 2002.

\bibitem[FGK05]{FGK05}
Joel Friedman, Andreas Goerdt, and Michael Krivelevich.
\newblock Recognizing more unsatisfiable random {$k$}-{SAT} instances
  efficiently.
\newblock {\em SIAM J. Comput.}, 35(2):408--430, 2005.

\bibitem[FO04]{FO04}
Uriel Feige and Eran Ofek.
\newblock Easily refutable subformulas of large random 3{CNF} formulas.
\newblock In {\em Proceedings of the 31st International Colloquium on Automata,
  Languages and Programming}, volume 3142 of {\em Lecture Notes in Comput.
  Sci.}, pages 519--530. Springer, Berlin, 2004.

\bibitem[FPV15]{FPV15}
Vitaly Feldman, Will Perkins, and Santosh Vempala.
\newblock {O}n the {C}omplexity of {R}andom {S}atisfiability {P}roblems with
  {P}lanted {S}olutions.
\newblock In {\em Proceedings of the 47th Annual ACM Symposium on Theory of
  Computing}, pages 77--86, 2015.

\bibitem[GB]{Gup11}
Anupam Gupta and Alex Beutel.
\newblock Lecture 12 - {S}emidefinite {D}uality.
\newblock Notes from course ``Linear and Semidefinite Programming".

\bibitem[GHP02]{GHP02}
Dima Grigoriev, Edward~A. Hirsch, and Dmitrii~V. Pasechnik.
\newblock Complexity of semi-algebraic proofs.
\newblock In {\em Proceedings of the 19th International Symposium on
  Theoretical Aspects of Computer Science}, pages 419--430, 2002.

\bibitem[Gri01]{Gri01}
Dima Grigoriev.
\newblock Linear lower bound on degrees of {P}ositivstellensatz calculus proofs
  for the parity.
\newblock {\em Theoretical Computer Science}, 259(1-2):613 -- 622, 2001.

\bibitem[KI06]{KI06}
Arist Kojevnikov and Dmitry Itsykson.
\newblock Lower {B}ounds of {S}tatic {L}ov{\'a}sz-{S}chrijver {C}alculus
  {P}roofs for {T}seitin {T}autologies.
\newblock In {\em Proceedings of the 33rd International Colloquium on Automata,
  Languages and Programming}, 2006.

\bibitem[LRS15]{LRS15}
James~R. Lee, Prasad Raghavendra, and David Steurer.
\newblock Lower bounds on the size of semidefinite programming relaxations.
\newblock In {\em Proceedings of the 47th Annual ACM Symposium on Theory of
  Computing}, pages 567--576, 2015.

\bibitem[LS91]{LS91}
L\'{a}szl\'{o} Lov\'{a}sz and Alexander Schrijver.
\newblock Cones of {M}atrices and {S}et-{F}unctions and 0-1 {O}ptimization.
\newblock {\em SIAM Journal on Optimization}, 1(2):166--190, 1991.

\bibitem[OW14]{OW14}
Ryan O'Donnell and David Witmer.
\newblock Goldreich's {PRG}: Evidence for near-optimal polynomial stretch.
\newblock In {\em Proceedings of the 29th Annual Conference on Computational
  Complexity}, pages 1--12, 2014.

\bibitem[Rag08]{Rag08}
Prasad Raghavendra.
\newblock Optimal {A}lgorithms and {I}napproximability {R}esults for {E}very
  {CSP}?
\newblock In {\em Proceedings of the 40th Annual ACM Symposium on Theory of
  Computing}, pages 245--254, 2008.

\bibitem[SA90]{SA90}
Hanif Sherali and Warren Adams.
\newblock A {H}ierarchy of {R}elaxations between the {C}ontinuous and {C}onvex
  {H}ull {R}epresentations for {Z}ero-{O}ne {P}rogramming {P}roblems.
\newblock {\em SIAM Journal on Discrete Mathematics}, 3(3):411--430, 1990.

\bibitem[Sch08]{Sch08}
Grant Schoenebeck.
\newblock {L}inear {L}evel {L}asserre {L}ower {B}ounds for {C}ertain
  $k$-{CSP}s.
\newblock In {\em Proceedings of the 49th Annual IEEE Symposium on Foundations
  of Computer Science}, pages 593--602, 2008.

\bibitem[TW13]{TW13}
Madhur Tulsiani and Pratik Worah.
\newblock ${LS}_+$ lower bounds from pairwise independence.
\newblock In {\em Proceedings of the 28th Annual Conference on Computational
  Complexity}, pages 121--132, 2013.

\bibitem[WJ08]{WJ08}
Martin~J. Wainwright and Michael~I. Jordan.
\newblock {\em Graphical Models, Exponential Families, and Variational
  Inference}, volume~1.
\newblock Now Publishers Inc., Hanover, MA, USA, January 2008.

\end{thebibliography}

\appendix
\section{Proofs from Section~\ref{sec:expansion}} \label{sec:expansion-proofs}
\begin{customfact}{\ref{fact:exp_to_bexp}}
$(s,k-d)$-expansion implies $(s,k-2d)$-boundary expansion.
\end{customfact}
\begin{proof}
Let $S$ be a set of at most $s$ hyperedges.  Each of the vertices in $\Gamma(S)$ is either a boundary vertex that appears in exactly one hyperedge or it appears in two or more hyperedges, so $|\Gamma(S)| \leq |\partial S| + \frac{1}{2}(|k|S| - |\partial S|)$.  Therefore, we can write
\[
|\partial S| \geq 2|\Gamma(S)| - k|S| \geq (k-2d)|S|,
\]
where the second inequality follows the expansion assumption.
\end{proof}

\begin{customlem}{\ref{lem:exp}}
Fix $\delta > 0$.  With high probability, a set of $m \leq n^{t/2-\eps}$ constraints chosen uniformly at random is both $\left(n^{\frac{\eps}{t-2}}, k-\frac{t}{2}+\frac{\eps}{2}\right)$-expanding and $\left(n^{\frac{\eps}{t-2}}, k-t+\eps\right)$-boundary expanding.
\end{customlem}
\begin{proof}
By Fact~\ref{fact:exp_to_bexp}, it suffices to show that a random instance is $\left(n^{\frac{\eps}{t-2}}, k-\frac{t}{2}+\frac{\eps}{2}\right)$-expanding.  We give the proof of \cite{OW14}, which is essentially the same as that of \cite{BGMT12}.

We want to upper bound the probability that any set of $r$ hyperdges with $r \leq n^{\frac{\eps}{t-2}}$ contains less than $r(k-\frac{t}{2}+\frac{\eps}{2})$ vertices.  Fix an $r$-tuple of edges $T$; this is a tuple of indices in $[m]$ representing the indices of the hyperedges in $T$.  We wish to upper bound $\Pr[|\Gamma(T)| \leq v]$; we can do this with the quantity
\[
\frac{(\text{\# sets $S$ of $v$ vertices}) \cdot (\text{\# sets of $r$ edges contained in $S$})}{(\text{\# of ways of choosing $r$ edges})}.
\]
Taking a union bound over all tuples of size $r$, we see that
\[
\Pr[|\Gamma(S)| \leq v~\forall S \text{ s.t. } |S| = r] \leq r!\binom{m}{r} \cdot \frac{\binom{n}{v} \binom{k!\binom{v}{k}}{r}}{(k!\binom{n}{k})^r}.
\]
Simplifying and applying standard approximations, we get that
\[
\Pr[|\Gamma(S)| \leq v~\forall S \text{ s.t. } |S| = r] \leq e^{(2+k)r + v}v^{kr-v}r^{-r}n^{v-kr}m^r.
\]
Set $v = \lfloor r(k-\frac{t}{2}+\frac{\eps}{2}) \rfloor$ and simplify to get
\[
\Pr\left[|\Gamma(S)| < r\left(k-\frac{t}{2}+\frac{\eps}{2}\right) ~\forall S \text{ s.t. } |S| = r\right] \leq (C(k,t) \cdot mn^{-(t/2-\eps/2)}r^{t/2-1-\eps/2})^r
\]
for some constant $C(k,t)$ depending on $k$ and $t$.
Then set $m = n^{t/2-\eps}$ and take a union bound over all choices of $r$ to get that
\begin{align*}
\Pr&\left[\text{$H_{\calI}$ not } \left(n^{\frac{\eps}{t-2}}\text{, } k-\frac{t}{2}+\frac{\eps}{2}\right)\text{-expanding}\right] \leq \sum_{r = 1}^{\lfloor n^{\eps/(t-2)}\rfloor} (C(k,t) \cdot n^{-\eps/2}r^{t/2-1-\eps/2})^r \\
& \qquad \qquad = \sum_{r = 1}^{\lceil \log n \rceil} (C(k,t) \cdot n^{-\eps/2}r^{t/2-1-\eps/2})^r + \sum_{r = \lceil \log n \rceil + 1}^{\lfloor n^{\eps/(t-2)}\rfloor} (C(k,t) \cdot n^{-\eps/2}r^{t/2-1-\eps/2})^r \\
& \qquad \qquad \leq 2 C(k,t) \cdot n^{-\eps/2}(\log n)^{t/2-\eps/2} + n^{\frac{\eps}{t-2}} (C(k,t) \cdot n^{-\eps/2}(n^{\frac{\eps}{t-2}})^{t/2-1-\eps/2})^{\log n} \\
& \qquad \qquad = O(n^{-\eps/3}). \qedhere
\end{align*}
\end{proof}

\section{Equivalence between PSDness of the degree-2 moment matrix and the covariance matrix}\label{apdx:Schur}
\begin{lemma} \label{lem:schur}
\begin{equation*}
\begin{pmatrix}
1& w^{\top}\\
w& B
\end{pmatrix}
\text{ is PSD } \iff
B-ww^{\top}
\text{ is PSD.\@}
\end{equation*}
\end{lemma}
\begin{proof}
\begin{align*}
\begin{pmatrix}
1& w^{\top}\\
w& B
\end{pmatrix}&
\text{ is PSD } \iff
\left(
(v_0\, v)
\begin{pmatrix}
1& w^{\top}\\
w& B
\end{pmatrix}
(v_0\, v)^{\top}
\ge 0
~\forall v_0 \in\mathbb{R},\, v\in\mathbb{R}^{nq}
\right)\\
&\iff
\left(
v_0^2 + 2\langle w, v\rangle v_0 + \langle Bv, v\rangle \ge 0
~\forall v_0 \in\mathbb{R},\, v\in\mathbb{R}^{nq}
\right)\\
&\iff
\left(
(v_0 + \langle w, v\rangle)^2 - \langle w, v\rangle^2 + \langle Bv, v\rangle \ge 0
~\forall  v_0 \in\mathbb{R},\, v\in\mathbb{R}^{nq}
\right)\\
&\iff
\left(
- \langle w, v\rangle^2 + \langle Bv, v\rangle \ge 0
~\forall  v\in\mathbb{R}^{nq}
\right)\\
&\iff
\left(
 v(B-ww^{\top})v^{\top} \geq 0
~\forall  v\in\mathbb{R}^{nq}
\right)\\
&\iff
B-ww^{\top}
\text{ is PSD.\@}
\end{align*}
\end{proof}

\begin{customlem}{\ref{lem:sa+-covariance}}
$M$ is PSD if and only if $\Sigma$ is PSD.
\end{customlem}
\begin{proof}
We rewrite $M$ as
\begin{equation*}
\begin{pmatrix}
1& w^{\top}\\
w& B
\end{pmatrix},
\end{equation*}
where $w$ is a vector whose $(i,a)$-element is $D_{\{i\}}(x_i = a)$ for $i\in[n]$ and $a\in[q]$
and $B$ is a matrix whose $((i,a), (j,b))$-element is $D_{\{i,j\}}(x_i = a \wedge x_j = b)$ for $i,j\in[n]$ and $a,b\in[q]$.  
From Lemma~\ref{lem:schur}, we know that
\begin{equation*}
\begin{pmatrix}
1& w^{\top}\\
w& B
\end{pmatrix}
\text{ is PSD if and only if }
B - ww^{\top}
\text{ is PSD.\@}
\end{equation*}
Observe that $B-ww^{\top}$ is equal to the covariance matrix $\Sigma$.
\end{proof}

\section{Proof of Lemma~\ref{lem:cond-local-cons}} \label{sec:cond-local-cons-pf}
\begin{customlem}{\ref{lem:cond-local-cons}}
Let $X \subseteq [n]$ and let $\{D_S\}$ be a family of $r$-locally consistent distributions for sets $S \subseteq [n]$ such that $S \cap X = \emptyset$ and $|S \cup X| \leq r$.   Then the family of conditional distributions $\{D_{S|X = \alpha}\}$ is $(r-|X|)$-locally consistent for any $\alpha \in \{0,1\}^X$ such that $D_X(X = \alpha) > 0$.
\end{customlem}
\begin{proof}
Tulsiani and Worah proved this lemma and we will use their proof \cite{TW13}.  Let $S \subseteq T$ and $|T \cup X| \leq r$.  Let $\beta$ be any assignment to $S$.  Then local consistency of the $\{D_S\}$ distributions implies that $D_{S \cup X}(S = \beta \wedge X = \alpha) = D_{T \cup X}(S = \beta \wedge X = \alpha)$ and $D_{S \cup X}(X = \alpha) = D_{T \cup X}(X = \alpha)$.  We therefore have that
\begin{align*}
D_{S|X = \alpha}(S = \beta) &= \frac{D_{S \cup X}(S = \beta \wedge X = \alpha)}{D_{S \cup X}(X = \alpha)} \\
&= \frac{D_{T \cup X}(S = \beta \wedge X = \alpha)}{D_{T \cup X}(X = \alpha)} = D_{T|X = \alpha}(S = \beta). \qedhere
\end{align*}
\end{proof}

\section{Proofs from Section~\ref{sec:cons-local-dists}} \label{sec:cons-local-proofs}
\begin{customlem}{\ref{lem:closure}}
If $H_{\calI}$ is $(s_1,e_1)$-expanding and $S$ is a set of variables such that $|S| < (e_1-e_2)s_1$ for some $e_2 \in (0,e_1)$, then there exists a set $\cl(S) \subseteq [n]$ such that $S \subseteq \cl(S)$ and $H_{\calI} - \cl(S)$ is $(s_2,e_2)$-expanding with $s_2 \geq s_1-\frac{|S|}{e_1-e_2}$ and $\cl(S) \leq \frac{e_1}{e_1-e_2}|S|$.
\end{customlem}
\begin{proof}
We compute $\cl(S)$ using the closure algorithm of \cite{BGMT12, TW13}:

\textbf{Input:} An $(s_1,e_1)$-expanding instance $\calI$, $e_2 \in (0,e_1)$, a tuple $S = (x_1,\ldots,x_u) \in [n]^u$ such that $u < (e_1-e_2)s_2$. \\
\textbf{Output:} The closure $\cl(S)$.

Set $\cl(S) \gets \emptyset$ and $s_2 \gets s_1$. \\
\textbf{for} $i = 1,\ldots,u$ 
\vspace{-1pt}
\begin{adjustwidth}{1cm}{}
$\cl(S) \gets \cl(S) \cup \{x_i\}$\\
\textbf{if} $H_{\calI} - \cl(S)$ is not $(s_2, e_2)$-expanding, \textbf{then} 
\vspace{-1pt}
\begin{adjustwidth}{1cm}{}
Find largest set of constraints $N_i$ in $H_{\calI} - \cl(S)$ such that $|N_i| \leq s_2$ and $|\Gamma(N_i)| \leq e_2|N_i|$.  \\Break ties by lexicographic order.\\
$\cl(S) \gets \cl(S) \cup \Gamma(N_i)$ \\
$s_2 \gets s_2 - |N_i|$
\end{adjustwidth}
\end{adjustwidth}
\textbf{return} $\cl(S)$

It is clear from the statement of the algorithm that $S \subseteq \cl(S)$.  We need to show that $H_{\calI} - \cl(S)$ is $(s_2,e_2)$-expanding, that $s_2 \geq s_1-\frac{|S|}{e_1-e_2}$, and that $\cl(S) \leq \frac{e_1}{e_1-e_2}|S|$.  We give the proof of \cite{BGMT12}.
\begin{enumerate}
\item $H_{\calI} - \cl(S)$ is $(s_2,e_2)$-expanding\\
We will show that $H_{\calI} - \cl(S)$ is $(s_2, e_2)$-expanding at every step of the algorithm.  Say we are in step $i$ and that $H_{\calI} - (\cl(S) \cup \{x_i\})$ is not $(s_2, e_2)$-expanding; if $H_{\calI} - (\cl(S) \cup \{x_i\})$ were $(s_2, e_2)$-expanding, we would be done.  Let $N_i$ be the largest set of hyperedges in $H_{\calI} - \cl(S)$ such that $|N_i| \leq s_2$ and $|\Gamma(N_i)| \leq e_2|N_i|$.  We need to show that $H_{\calI} - (\cl(S) \cup \{x_i\} \cup \Gamma(N_i))$ is $(s_2 - |N_i|, e_2)$-expanding.

To see this, assume for a contradiction that there exists a set of hyperedges $N'$ in $H_{\calI} - (\cl(S) \cup \{x_i\} \cup \Gamma(N_i))$ such that $N' \leq s_2 - |N_i|$ and $|\Gamma(N')| < e_2 |N'|$.  Consider $N_i \cup N'$.  Note that $N_i$ and $N'$ are disjoint, so $|N_i \cup N'| \leq s_2$.  Also, $|\Gamma(N_i \cup N')| \leq e_2|N_i| + e_2|N'| = e_2|N_i \cup N'|$.  This contradicts the maximality of $N_i$.

\item $s_2 \geq s_1-\frac{|S|}{e_1-e_2}$\\
Consider the set $N = \bigcup_{i = 1}^u N_i$.  First, note that $|N| = s_1 - s_2$, so $|\Gamma(N)| \geq e_1(s_1-s_2)$ by expansion of $H_{\calI}$.  Second, each element of $\Gamma(N) - S$ occurs in exactly one of the $N_i$'s and each $N_i$ has expansion at most $e_2$.  Using these two observations, we see that
\[
e_1(s_1-s_2) \leq |\Gamma(N)| \leq |S| + \sum_{i=1}^u e_2 |N_i| = |S| + e_2(s_1-s_2).
\]
This implies the claim.

\item $\cl(S) \leq \frac{e_1}{e_1-e_2}|S|$\\
Observe that $\cl(S) = S \cup \bigcup_{i=1}^u \Gamma(N_i)$.  Also, every $N_i$ has expansion at most $e_2$.  Therefore, we have that
\begin{align*}
|\cl(S)| &\leq |S| + \sum_{i=1}^u |\Gamma(N_i)| \\
&\leq |S| + e_2 \sum_{i=1}^u |N_i| \\
&\leq |S| + \frac{e_2|S|}{e_1-e_2} \\
&= \left(\frac{e_1}{e_1-e_2}\right)|S|,
\end{align*}
where we used that $\sum_{i=1}^u |N_i| = s_1-s_2$ and $s_2 \geq s_1-\frac{|S|}{e_1-e_2}$. \qedhere
\end{enumerate}
\end{proof}

\begin{customtheorem}{\ref{thm:local-cons}}
For a random instance $\calI$ with $m \leq \Omega(n^{t/2-\eps})$, the family of distributions $\{\Pi_S\}_{|S| \leq r}$ is $r$-locally consistent for $r = n^{\frac{\eps}{t-2}}$ and is supported on satisfying assignments.
\end{customtheorem}
To prove the theorem, we will use the following lemma, which says that the local distributions $\Pi'_S$ and $\Pi'_T$ with $S \subseteq T$ are consistent if $H_{\calI} - S$ has high boundary expansion.
\begin{lemma}{\label{lem:local-cons}}
Let $P$ be a $(t-1)$-wise uniform supporting predicate, let $\calI$ be an instance of $\CSP(P)$, and let $S \subseteq T$ be sets of variables.  If $H_{\calI}$ and $H_{\calI} - S$ are $(r, k-t+\eps)$-boundary expanding for some $\eps > 0$ and $\calC(T) \leq r$, then for any $\alpha \in [q]^S$, $\Pi'_S(S = \alpha) = \Pi'_T(S = \alpha)$.
\end{lemma}
First, we will use this lemma to prove Theorem~\ref{thm:local-cons}.
\begin{proof}[Proof of Theorem~\ref{thm:local-cons}]
Let $S \subseteq T$ be sets of variables with $|T| \leq r$.  Consider $U = \cl(S) \cup \cl(T)$.  We will show that both $\Pi_S$ and $\Pi_T$ are consistent with $U$ and therefore must themselves be consistent.  Observe that $|\cl(S)|$ and $|\cl(T)|$ are at most $\frac{2kr}{\eps}$, so $|U| \leq \frac{4kr}{\eps}$.  We want to apply Lemma~\ref{lem:local-cons}, so we will first show that $|\calC(U)| \leq \frac{8r}{\eps}$.  Assume for a contradiction that $C$ is a subset of $\calC(U)$ of size $\frac{8r}{\eps}$.  Then
\[
\frac{|\Gamma(C)|}{|C|} \leq \frac{|U|}{|C|} = \frac{4kr/\eps}{8r/\eps} = \frac{k}{2} < k -\frac{t}{2} + \frac{\eps}{2},
\]
which violates expansion (Lemma~\ref{lem:exp}).

We know that $H_{\calI} - \cl(T)$ and $H_{\calI} - \cl(S)$ are $(r, k-t+\eps)$-boundary expanding for some $\eps > 0$.  We can then apply Lemma~\ref{lem:local-cons} twice with sets $\cl(S) \subseteq U$ and $\cl(T) \subseteq U$ to see that
\[
\Pi_S(S = \alpha) = \Pi_{\cl(S)}'(S = \alpha) = \Pi_{U}'(S = \alpha) = \Pi'_{\cl(T)}(S = \alpha) =\Pi_T(S = \alpha). \qedhere
\]
\end{proof}
Now we prove Lemma~\ref{lem:local-cons}.
\begin{proof}[Proof of Lemma~\ref{lem:local-cons}]
We follow the proof of Benabbas et al. \cite{BGMT12}.  Let $\calC(T) \setminus \calC(S) = \{C_1,\ldots,C_u\}$ and, for a constraint $C$, let $\sigma(C)$ be the variables in the support of $C$.  First, observe that
\begin{align*}
Z_T \sum_{\substack{\beta \in [q]^T \\ \beta_S = \alpha}} \Pi'_{T}(\beta) &= \sum_{\gamma \in [q]^{T \setminus S}} \prod_{C \in \calC(T)} \mu_C((\alpha, \gamma)) \\
&= \left(\prod_{C \in \calC(S)} \mu_C(\alpha)\right) \sum_{\gamma \in [q]^{S \setminus T}} \prod_{i = 1}^u \mu_{C_i}((\alpha,\gamma)) \\
&= (Z_S \Pi'_S(\alpha))\sum_{\gamma \in [q]^{S \setminus T}} \prod_{i = 1}^u \mu_{C_i}((\alpha,\gamma))
\end{align*}
To finish the proof, we need the following claim.
\begin{claim} \label{cl:php}
There exists an ordering $(C_{i_1},\ldots,C_{i_u})$ of constraints of $\calC(T) \setminus \calC(S)$ and a partition $V_1,\cdots,V_u,V_{u + 1}$ of variables of $T \setminus S$ such that for all $j \leq u$ the following hold.
\begin{enumerate}
\item $V_j \subseteq \sigma(C_{i_j})$.
\item $|V_j| \geq k-t+1$
\item $V_j$ does not intersect $\sigma(C_{i_{l}})$ for any $l > j$.  That is, $V_j \cap \bigcup_{l > j} \sigma(C_{i_l}) = \emptyset$.
\end{enumerate}
\end{claim}
\begin{proof}[Proof of Claim~\ref{cl:php}]
We will find the sets $V_j$ by repeatedly using $(r, k-t+\eps)$-boundary expansion of $H_{\calI} - S$.  Let $Q_1 = \calC(T) \setminus \calC(S)$.  We know that $|Q_1| \leq r$, so boundary expansion of $H_{\calI} - S$ implies that $|\partial(Q_1) \setminus S| \geq (k-t+\eps)|Q_1|$.  There must exist a constraint $C_j \in Q_1$ with at least $k-t+1$ boundary variables in $H_{\calI} - S$; i.e., $|\sigma(C_j) \cap (\partial(Q_1) \setminus S)| \geq k-t+1$.  We then set $V_1 = \sigma(C_j) \cap (\partial(Q_1) \setminus S)$ and $i_1 = j$.  Let $Q_2 = Q_1 \setminus C_j$.  We apply the same process $u-1$ more times until $Q_l$ is empty and then set $V_{u+1} = (T \setminus S) \setminus (\bigcup_{j=1}^u V_j)$.  We remove constraint $C_{i_l}$ at every step and $V_l \subseteq \sigma(C_{i_l})$, so it holds that $V_j \cap \bigcup_{l > j} \sigma(C_{i_l}) = \emptyset$.
\end{proof}
Using the claim, we can write $\sum_{\gamma \in [q]^{S \setminus T}} \prod_{i = 1}^u \mu_{C_i}((\alpha,\gamma))$ as
\[
\sum_{\gamma_{u+1} \in [q]^{V_{u+1}}} \sum_{\gamma_u \in [q]^{V_u}} \mu_{C_u}(\gamma'_u) \sum_{\gamma_{u-1} \in [q]^{V_{u-1}}} \mu_{C_{u-1}}(\gamma'_{u-1}) \cdots  \sum_{\gamma_1 \in [q]^{V_1}} \mu_{C_1}(\gamma'_1),
\]
where each $\gamma'_j$ depends on $\alpha$ and $\gamma_l$ with $l \geq j$ but does not depend on $\gamma_l$ with $l < j$.  We will evaluate this sum from right to left.  We know that each $V_j$ contains at least $k-t+1$ elements, so $(t-1)$-wise uniformity of $\mu$ implies that $\sum_{\gamma_j \in [q]^{V_j}} \mu_{C_j}(\gamma'_j) = q^{-(k-|V_j|)}$.  Applying this repeatedly, we see that
\[
\sum_{\gamma \in [q]^{S \setminus T}} \prod_{i = 1}^u \mu_{C_i}((\alpha,\gamma)) = q^{-(ku - \sum_{j=1}^{u+1} |V_j|)} = q^{|T \setminus S| - k|\calC(T) \setminus \calC(S)|}.
\]
Plugging this quantity into the above calculation, we obtain
\[
Z_T \sum_{\substack{\beta \in [q]^T \\ \beta|_S = \alpha}} \Pi'_{T}(\beta) = Z_S \Pi'_S(\alpha) q^{|T \setminus S| - k|\calC(T) \setminus \calC(S)|}.
\]
Since $H_{\calI}$ has $(r,k-t+\eps)$-boundary expansion for some $\eps > 0$, we can set $S = \emptyset$ to get that $Z_T = q^{|T|-k|\calC(T)|}$. Similarly, $Z_S = q^{|S|-k|\calC(S)|}$.  Plugging these two quantities in completes the proof.
\end{proof}

\section{Equivalence of Sherali-Adams, \SAplus, and static \LSplus tightenings of linear and degree-$k$ relaxations of $\CSP(P)$} \label{sec:diff-encodings}
\begin{customlem}{\ref{lem:diff-encodings}}
Let $r \geq k$ and let $\calI$ be an instance of CSP$(P)$ with binary alphabet.   Then the following statements hold.
\begin{enumerate}
\item $\SAop^r(R_{\calI}) \subseteq \SAop^{r+k+1}(L_{\calI})$ and $\SAop^r(L_{\calI}) \subseteq \SAop^{r+k+1}(R_{\calI})$.
\item $\SAplusop^r(R_{\calI}) \subseteq \SAplusop^{r+k+1}(L_{\calI})$ and $\SAplusop^r(L_{\calI}) \subseteq \SAplusop^{r+k+1}(R_{\calI})$.
\item $\mathrm{StaticLS}_+^r(R_{\calI}) \subseteq \mathrm{StaticLS}_+^{r+k+1}(L_{\calI})$ and $\mathrm{StaticLS}_+^r(L_{\calI}) \subseteq \mathrm{StaticLS}_+^{r+k+1}(R_{\calI})$.
\end{enumerate}
\end{customlem}
\begin{proof}
First, we recall some notation from Section~\ref{sec:prelims}.  Let $P'(x)$ be the unique degree-$k$ polynomial such that $P'(z) = P(z)$ for all $z \in \{0,1\}^k$; assume $P$ and $P'$ depend on all $k$ of their input variables.  Let $F = \{z \in \{0,1\}^k \suchthat P(z)=0\}$.   For $b \in \{0,1\}$, define $a^{(b)}$ so that $a^{(b)}$ is $a$ if $b = 0$ and $1-a$ if $b=1$. For $z \in [0,1]^k$ and $c \in \{0,1\}^k$, we define $z^{(c)} \in [0,1]^k$ so that $(z^{(c)})_i = z_i^{(c_i)}$.  For $f \in \{0,1\}^k$ and $z \in [0,1]^k$, let $P_{f}(z) = \sum_{i = 1}^k z^{(f_i)}$.  Let $(c,S) \in \calI$ be any constraint.  Note that
\begin{equation} \label{eq:predicate-as-sum-of-linear-cons}
P'(x_S^{(c)}) - 1 = \sum_{f \in F} \indic{\{x^{(c)}_S = f\}}(x) \cdot (P_{f}(f)-1) = -\sum_{f \in F} \indic{\{x^{(c)}_S = f\}}(x).
\end{equation}

We give the proof for \SA.  The \SAplus and static \LSplus cases are identical, as constraints in \SAplus and static \LSplus generate exactly the same lifted constraints as in \SA.

First, assume that we have a family of $(r+k+1)$-locally consistent distributions $\{D_S\}$ satisfying
\begin{equation} \label{eq:sa-deg-1-sat}
\E_D[\indic{\{x_T = \alpha\}}(x) \cdot (P_{f}(x^{(c)}_S)-1)] \geq 0
\end{equation}
for all $f \in F$, $T \subseteq [n]$, and $\alpha \in \{0,1\}^{|T|}$ such that $\deg(\indic{\{x_T = \alpha\}}(x) \cdot (P'_{f}(x^{(c)}_S)-1)) \leq r+k+1$.
Fix any $U \subseteq [n]$ and $\beta \in \{0,1\}^{|U|}$ such that $\deg(\indic{\{x_U = \beta\}}(x) \cdot (P'(x_S^{(c)}) - 1)) \leq r$.

We want to show that $\E_D[\indic{\{x_U = \beta\}}(x) \cdot (P'(x_S^{(c)}) - 1)] = 0$.  Using \eqref{eq:predicate-as-sum-of-linear-cons}, it suffices to show that
\[
\E_D[\indic{\{x_U = \beta\}}(x) \cdot \indic{\{x^{(c)}_S = f\}}(x)] = 0
\]
for all $f \in F$.
First, we need to bound the size of $|U|$.  Again using \eqref{eq:predicate-as-sum-of-linear-cons}, we know that
\begin{align*}
\indic{\{x_U = \beta\}}(x) \cdot (P'(x_S^{(c)}) - 1) &= -\indic{\{x_U = \beta\}}(x) \sum_{f \in F} \indic{\{x^{(c)}_S = f\}}(x) \\
&= -\indic{\{x_{U \setminus S} = \beta_{U \setminus S}\}}(x) \sum_{f \in F} \indic{\{x_{S \cap U} = \beta_{S \cap U}\}}(x) \cdot \indic{\{x^{(c)}_S = f\}}(x).
\end{align*}
We have two cases.  If the assignments $x_{S \cap U} = \beta_{S \cap U}$ and $x^{(c)}_S = f$ are inconsistent for all $f \in F$, then $\indic{\{x_U = \beta\}}(x) \cdot \indic{\{x^{(c)}_S = f\}}(x) = 0$ for all $f \in F$ after multilinearization.  Then $\E_D[\indic{\{x_U = \beta\}}(x) \cdot \indic{\{x^{(c)}_S = f\}}(x)] = 0$ for all $f \in F$ and we are done.

Otherwise,
\[
\sum_{f \in F} \indic{\{x_{S \cap U} = \beta_{S \cap U}\}}(x) \cdot \indic{\{x^{(c)}_S = f\}}(x) \ne 0 
\]
and $\deg(\indic{\{x_U = \beta\}}(x) \cdot (P'(x_S^{(c)}) - 1)) \leq r$ implies that $\deg(\indic{\{x_{U \setminus S} = \beta_{U \setminus S}\}}(x)) \leq r$.  Then $|U \setminus S| \leq r$ and $|U \cup S| \leq r+k$.  Since $\indic{\{x_U = \beta\}}(x) \cdot \indic{\{x^{(c)}_S = f\}}(x) \geq 0$ for all $x \in \{0,1\}^n$ and $\indic{\{x_U = \beta\}}(x) \cdot \indic{\{x^{(c)}_S = f\}}(x)$ depends on at most $r+k$ variables, we know that
\begin{equation} \label{eq:pe-lb}
\E_D[\indic{\{x_U = \beta\}}(x) \cdot \indic{\{x^{(c)}_S = f\}}(x)] \geq 0.
\end{equation}
On the other hand, assumption \eqref{eq:sa-deg-1-sat} implies that
\[
-\E_D[\indic{\{x_U = \beta\}}(x) \cdot \indic{\{x^{(c)}_S = f\}}(x)] = \E_D[\indic{\{x_U = \beta\}}(x) \cdot \indic{\{x^{(c)}_S = f\}}(x) \cdot (P_{f}(x^{(c)}_S)-1)] \geq 0
\]
since $\deg(\indic{\{x_U = \beta\}}(x) \cdot \indic{\{x^{(c)}_S = f\}}(x) \cdot (P_{f}(x^{(c)}_S)-1)) \leq r+k+1$.

For the other direction, assume that we have a family of $(r+k+1)$-locally consistent distributions $\{D_S\}$ satisfying
\begin{equation} \label{eq:sa-deg-k-sat}
\E_D[\indic{\{x_T = \alpha\}} \cdot (P'(x_S^{(c)}) - 1)] = 0
\end{equation}
for all $T \subseteq [n]$ and $\alpha \in \{0,1\}^{|T|}$ such that $\deg(\indic{\{x_T = \alpha\}} \cdot (P'(x_S^{(c)}) - 1)) \leq r+k+1$.  Fix any $U \subseteq [n]$ and $\beta \in \{0,1\}^{|U|}$ such that $\deg(\indic{\{x_U = \beta\}}(x) \cdot (P_f(x_S^{(c)}) - 1)) \leq r$.

We want to show that $\E_D[\indic{\{x_U = \beta\}}(x) \cdot (P_f'(x_S^{(c)}) - 1)] \geq 0$.  We will do this by proving that
\[
\E_D[\indic{\{x^{(c)}_S = z\}}(x) \cdot \indic{\{x_U = \beta\}}(x) \cdot (P_{f}(x^{(c)}_S)-1)] \geq 0
\]
for any $z \in \{0,1\}^k$ and then summing over all~$z$.

If $S \subseteq U$ and $\beta$ assigns $S$ to $f$, then $\indic{\{x_U = \beta\}}(x) \cdot (P_f(x_S^{(c)}) - 1) = 0$ after multilinearizing and we are done.  Otherwise, it is easy to see that $|U| \leq r$.

We consider two cases: $z = f$ and $z \ne f$.  In the first case, we can use \eqref{eq:sa-deg-k-sat} and \eqref{eq:predicate-as-sum-of-linear-cons} to see that
\[
\sum_{f \in F} \E_D[\indic{\{x_U = \beta\}}(x) \cdot \indic{\{x^{(c)}_S = f\}}(x)] = 0.
\]
By \eqref{eq:pe-lb}, each term in the sum must be $0$, so we have that
\[
\E_D[\indic{\{x_U = \beta\}}(x) \cdot \indic{\{x^{(c)}_S = f\}}(x) \cdot (P_{f}(x^{(c)}_S)-1)] = -\E_D[\indic{\{x_U = \beta\}}(x) \cdot \indic{\{x^{(c)}_S = z\}}(x)]=0.
\]

When $z \ne f$, \eqref{eq:pe-lb} and the fact that $P_{f}(z)-1 \geq 0$ imply that
\[
\E_D[\indic{\{x_U = \beta\}}(x) \cdot \indic{\{x^{(c)}_S = z\}}(x) \cdot (P_{f}(x^{(c)}_S)-1)] = (P_{f}(z)-1) \cdot \E_D[\indic{\{x_U = \beta\}}(x) \cdot \indic{\{x^{(c)}_S = z\}}(x)] \geq 0. \qedhere
\]
\end{proof}

\section{Correspondence between static \LSplus proof system and SDP relaxation} \label{sec:static-ls}
Say we start with a set of constraints $A = \{g_1(x) \geq 0, g_2(x) \geq 0, \ldots, g_m(x) \geq 0\}$.  Recall that an $r$-round static \LSplus solution is a set of local distributions $\{D_S\}$ satisfying the following conditions.
\begin{enumerate}
\item $\{D_S\}_{S\subseteq [n],\, |S|\le r}$ is $r$-locally consistent. \label{enum:static-ls+-local-cons}
\item $\E_{x \sim D}[\indic{\{x_T = \alpha\}}(x) \cdot g(x)] \geq 0$ for all $g \in A$, $T \subseteq [n]$, and $\alpha \in [q]^{|T|}$ such that $\deg(\indic{\{x_T = \alpha\}}(x) \cdot f(x)) \leq r$. \label{enum:static-ls+-sat} 
\item[3$'$.] $M_{X,\alpha}$ is PSD for all $X \subseteq [n]$ with $|X| \leq r-2$ and all $\alpha \in [q]^{X}$. \manuallabel{enum:static-ls+-psd}{3$'$}
\end{enumerate}

A static \LSplus refutation has the form
\begin{equation} \label{eq:static-lsplus-refutation}
\sum_i \gamma_i \cdot b_i(x) \cdot \indic{\{x_{T_i} = \alpha_i\}}(x) + \sum_j (x_j^2-x_j) h_j(x) = -1,
\end{equation}
where $\gamma_i \geq 0$, $b_i$ is an axiom or the square of an affine function, and the $h_j$'s are arbitrary polynomials.

\begin{proposition} \label{prop:static-ls+-sdp-iff-ref}
The $r$-round static \LSplus SDP is infeasible if and only if a degree-$r$ static \LSplus refutation exists.
\end{proposition}
\begin{proof}

Using linearity of $\E[\cdot]$, we can write the first two \SA conditions as linear constraints in the variables $\E_D[\indic{\{x_S = \beta\}}(x)]$.  The final constraint requires that the matrices $M_{T,\alpha}$, whose entries are also variables of the form $\E_D[\indic{\{x_S = \beta\}}(x)]$, is PSD.   As mentioned above, we can arrange the matrices $M_{T,\alpha}$ into a block diagonal matrix $\mathcal{M}$ such that $\mathcal{M}$ is PSD if and only if each of the $M_{T,\alpha}$'s are PSD.  Let $d$ be the dimension of $\mathcal{M}$.  We can think of the $r$-round \SA constraints as being linear constraints on the entries of $\mathcal{M}$.  In particular, say these linear \SA constraints have the form $A \cdot \mathrm{vec}(\mathcal{M}) \geq b$, where $\mathrm{vec}(\mathcal{M}) \in \R^{d^2}$ is the vector formed by concatenating the columns of $\mathcal{M}$.  Let $c$ be the number of rows of $A$.  Then we can write the static \LSplus SDP as
\begin{align*}
A \cdot \mathrm{vec}(\mathcal{M}) \geq b \\
\mathcal{M} \succeq 0.
\end{align*}

First, we show that the existence of a degree-$r$ static \LSplus refutation implies that the $r$-round static \LSplus SDP is infeasible.  Assume for a contradiction that there exists a family of local distributions $\{D_S\}$ satisfying the three constraints above.  We will derive a contradiction by applying $E_D[\cdot]$ to each term of \eqref{eq:static-lsplus-refutation}.  Specifically, we will show that if $\deg(b_i(x) \cdot \indic{\{x_{T_i} = \alpha_i\}}(x)) \leq r$, then $\E_D[\gamma_i \cdot b_i(x) \cdot \indic{\{x_{T_i} = \alpha_i\}}(x)] \geq 0$.  Applying $\E_D[\cdot]$ to the left-hand side of \eqref{eq:static-lsplus-refutation} gives value at least $0$, applying $\E_D[\cdot]$ to the right-hand side gives value $-1$, and we obtain a contradiction.  To show that $\E_D[\gamma_i \cdot b_i(x) \cdot \indic{\{x_{T_i} = \alpha_i\}}(x)] \geq 0$, we will consider two cases.
\paragraph{Case 1:}  $b_i$ is an axiom.\\
This case is immediate from Condition~\ref{enum:static-ls+-sat}.

\paragraph{Case 2:} $b_i$ is the square of a linear form.\\
This case follows almost immediately from Condition~\ref{enum:static-ls+-psd}.  Write $b_i(x)$ as follows:
\[
b_i(x) = \left(a_0 + \sum_{u \in [n]} a_u x_u\right)^2 = \sum_{u,v \in [n]} a_u a_v x_u x_v + 2a_0  \sum_{u \in [n]} a_u x_u + a_0^2.
\]
Then we have the following calculation:
\begin{align*}
\E_D[b_i(x) \cdot \indic{\{x_{T_i} = \alpha_i\}}(x)] &= \sum_{u,v \in [n]} a_u a_v \cdot \E_D[x_u x_v \cdot \indic{\{x_{T_i} = \alpha_i\}}] + 2a_0  \sum_{u \in [n]} a_u \cdot \E_D[x_u \cdot \indic{\{x_{T_i} = \alpha_i\}}]  \\
& \qquad \qquad + a_0^2 \cdot \E_D[\indic{\{x_{T_i} = \alpha_i\}}]\\
&= \sum_{u,v \in [n]} a_u a_v \cdot M_{T_i,\alpha_i}((u,1),(v,1))\\
& \qquad \qquad+ 2a_0  \sum_{u \in [n]} a_u \cdot M_{T_i,\alpha_i}((u,1),0)  + a_0^2 \cdot M_{T_i,\alpha_i}(0,0) \\
&= (a')^{\top} (M_{T_i,\alpha_i}) a' \quad \text{where $a'(u,1) = a_u$ and $a'(u,0) = 0$ for all $u \in [n]$ and $a'(0) = a_0$}\\
& \geq 0 \qquad \text{by Condition~\ref{enum:static-ls+-psd}}.
\end{align*}

For the other direction, assume that the $r$-round static \LSplus SDP is infeasible; we want to prove that a refutation \eqref{eq:static-lsplus-refutation} exists.  Assume there exists an \SA pseudoexpectation satisfying Conditions~\ref{enum:static-ls+-local-cons} and \ref{enum:static-ls+-sat}.  Otherwise, we can find an \SA refutation, which is also a valid static \LSplus refutation.
Then the sets $\{\mathcal{M} \in \R^{d \times d} : A \cdot \mathrm{vec}(\mathcal{M}) \geq b\}$ and $\{\mathcal{M} \in \R^{d \times d} : \text{$\mathcal{M}$ is PSD}\}$ are both nonempty, but their intersection is empty.  Let $A \bullet B = \sum_{ij} A_{ij} B_{ij}$.  We will need the following claim.
\begin{claim}
Let $S \subseteq \R^{d \times d}$ be convex, closed, and bounded.  Suppose that for all $X \in S$, $X$ is not PSD.  Then there exists a PSD matrix $C \in \R^{d \times d}$ such that $C \bullet X < 0$ for all $X \in S$.
\end{claim}
\begin{proof}[Proof of Claim]
The claim follows from the following two results.
\begin{theorem}[Separating Hyperplane Theorem]
Let $S, T \subseteq \R^d$ be closed, convex sets such that $S \cap T = \emptyset$ and $S$ is bounded.  Then there exist $a \ne 0$ and $b$ such that
\[
a^{\top} x > b \text{ for all $x \in S$ and } a^{\top} x \leq b \text{ for all $x \in T$.}
\]
\end{theorem}
\begin{lemma}[\textup{e.g., \cite[Lemma~12.4]{Gup11}}]\label{lem:psdness-dot-product}
If $A \bullet B \geq 0$ for all PSD $B$, then $A$ is PSD.
\end{lemma}
Applied to our situation, the Separating Hyperplane Theorem says that there exists $C$ and $\delta$ such that $C \bullet X < \delta$ for all $X \in S$ and $C \bullet X \geq \delta$ for all PSD $X$.  We need to show that we can choose $\delta = 0$.  Applying Lemma~\ref{lem:psdness-dot-product} will then complete the proof.

We know $\delta \leq 0$ because the zero matrix is PSD. It remains to show that we can choose $\delta \geq 0$.  Assume for a contradiction that there exists PSD $X$ such that $C \bullet X < 0$.  We can then scale $X$ by a large enough positive constant to get a PSD matrix $X'$ such that $C \bullet X' < \delta$, a contradiction.
\end{proof}

The claim implies that there is a PSD matrix $C$ such that the set
\[
\{\mathcal{M} \in \R^{d \times d} : A \cdot \mathrm{vec}(\mathcal{M}) \geq b, C \bullet \mathcal{M} \geq 0\}
\]
is empty.  As this set is defined by linear inequalities, we can apply Farkas' Lemma.
\begin{theorem}[Farkas' Lemma]
Let $A \in \R^{m \timesÊn}$ and consider a system of linear inequalities $Ax \geq b$.  Exactly one of the following is true.
\begin{enumerate}
\item There is an $x \in \R^n$ such that $Ax \geq b$.
\item There is a $y \in \R^m$ such that $y \geq 0$, $y^{\top} A = 0$, and $y^{\top} b > 0$.
\end{enumerate}
\end{theorem}
In particular, this implies that there exist $y \in \R^c$ and $z \in \R$ such that $y \geq 0$, $z \geq 0$, and
\begin{equation} \label{eq:farkas}
y^{\top}(A \cdot \mathrm{vec}(\mathcal{M}) - b) + z C \bullet \mathcal{M} < 0
\end{equation}
for all $\mathcal{M} \in \R^{d \times d}$.  Since $C$ is PSD, we can write its eigendecomposition $C = \sum_{\ell} \lambda_{\ell} v_{\ell} v_{\ell}^{\top}$ with $\lambda_{\ell} \geq 0$ for all $\ell$.  Also, recall that $\mathcal{M}$ is block diagonal with blocks $M_{T,\alpha}$.  This block structure induces a corresponding partition of $[d]$.  We can write the vector $v_{\ell} \in \R^d$ as $(v_{\ell,T,\alpha})_{T,\alpha}$ using this partition.  Then the second term of \eqref{eq:farkas} is
\begin{align*}
z C \bullet X &= z \sum_{\ell} \lambda_{\ell} (v_{\ell} v_{\ell}^{\top}) \cdot \mathcal{M} \\
&= z \sum_{\ell} \lambda_{\ell} v_{\ell}^{\top} \mathcal{M} v_{\ell} \\
&= z \sum_{\ell} \lambda_{\ell} \sum_{\substack{|T| \leq r-2 \\ \alpha \in \{0,1\}^T}} v_{\ell,T,\alpha}^{\top} M_{T,\alpha} v_{\ell,T,\alpha} \\
&= z \sum_{\ell} \lambda_{\ell} \sum_{\substack{|T| \leq r-2 \\ \alpha \in \{0,1\}^T}} \sum_{i,j \in [n]} v_{\ell,T,\alpha}(i) v_{\ell,T,\alpha}(j) M_{T,\alpha}(i,j).
\end{align*}
Overall, we get
\[
y^{\top}(A \cdot \mathrm{vec}(M) - b) + z \sum_{\ell} \lambda_{\ell} \sum_{\substack{|T| \leq r-2 \\ \alpha \in \{0,1\}^T}} \sum_{i,j \in [n]} v_{\ell,T,\alpha}(i) v_{\ell,T,\alpha}(j) M_{T,\alpha}(i,j) < 0
\]
for all $\mathcal{M} \in \R^{d \times d}$.  Finally, we substitute in indicator polynomials $\indic{\{x_T = \alpha\}}(x)$ for the entries of $\mathcal{M}$ and scale appropriately to get an \LSplus refutation of the form \eqref{eq:static-lsplus-refutation}.  For the first term, since each row $a_i$ of $A$ each $b_i$ correspond to an \SA constraint $a_i \cdot \mathrm{vec}(M) \geq b_i$, this substitution gives an expression of the form
\[
\sum_i y_i \cdot g_i(x) \cdot \indic{\{x_{T_i} = \alpha_i\}}(x),
\]
where each $g_i(x) \geq 0$ is one of our initial constraints and $y_i \geq 0$.  The second term has the form
\[
z \sum_{\ell} \lambda_{\ell} \sum_{\substack{|T| \leq r-2 \\ \alpha \in \{0,1\}^T}} \left(\sum_{i \in [n]} v_{\ell,T,\alpha}(i) \cdot x_i\right)^2 \indic{\{x_{T_i} = \alpha_i\}}(x). \qedhere
\]
\end{proof}

\end{document}